\documentclass[twoside,11pt]{article}

\usepackage{jmlr2e}

\usepackage[symbol]{footmisc}
\usepackage{epstopdf,gensymb}
\usepackage{graphics}
\usepackage{graphicx,color}
\newcommand{\erdos}{Erd\H{o}s-R\'enyi }
\usepackage{amsfonts, animate,subfigure}
\usepackage{amssymb}
\usepackage{hyperref}
\usepackage{setspace}
\newcommand{\Jopt}{\left(\widehat B_{opt}^{(n)}, \widehat L_{opt}^{(n)}\right)}
\usepackage{indentfirst}
\usepackage{float}
\usepackage{amsfonts}
\usepackage{bm,pbox}
\usepackage{multirow} 
\usepackage{color}
\usepackage{framed}
\usepackage{amssymb,amsfonts,amsbsy,latexsym,amsmath}
\usepackage{enumerate}
\usepackage{verbatim}
\newenvironment{enumeratei}{\begin{enumerate}[\upshape (i)]}{\end{enumerate}}

\newtheorem{theorem}{Theorem}
\newtheorem{lemma}[theorem]{Lemma}

\newtheorem{corollary}[theorem]{Corollary}
\newtheorem{definition}[theorem]{Definition}

\renewcommand{\leq}{\leqslant} 
\renewcommand{\geq}{\geqslant}

\newcommand{\eps}{\varepsilon}

\newcommand{\set}[1]{\left\{#1\right\}}

\def\qed{ \hfill $\blacksquare$}  

\newcommand{\cB}{\mathcal{B}}
\newcommand{\cD}{\mathcal{D}}
\newcommand{\cI}{\mathcal{I}}

\newcommand{\cR}{\mathcal{R}}

\newcommand{\cZ}{\mathcal{Z}}  

\newcommand{\cH}{H_\ast}
\newcommand{\cQ}{q}

\newcommand{\mvd}{\boldsymbol{d}}

\newcommand{\expv}{\mathbb{E}}

\newcommand{\halfk}{\frac{1}{2^k}}

\newcommand{\halfkMO}{\frac{1}{2^{k-1}}}
\newcommand{\Bopt}{\widehat B_{opt}(n)}

\DeclareMathOperator*{\argmax}{arg\,max}

\newcommand{\SmallNa}{N_{n,k}(A)}
\newcommand{\SmallNap}{N_{n,k+1}(A')}
\newcommand{\Bv}{v}

\jmlrheading{18}{2017}{**}{12/16; Revised 6/17}{11/17}{James D. Wilson, John Palowitch, Shankar Bhamidi and Andrew B. Nobel}

\ShortHeadings{Community Extraction in Multilayer Networks}{Wilson, Palowitch, Bhamidi and Nobel}
\firstpageno{1}

\begin{document}

\title{Community Extraction in Multilayer Networks with Heterogeneous Community Structure}

\author{\name James D.\ Wilson\footnotemark \email jdwilson4@usfca.edu \\
       \addr Department of Mathematics and Statistics\\
       University of San Francisco\\
       San Francisco, CA 94117-1080
       \AND
       \name John Palowitch$^*$ \email palojj@email.unc.edu \\
	   \name Shankar Bhamidi \email bhamidi@email.unc.edu\\
	   \name Andrew B.\ Nobel \email nobel@email.unc.edu\\
       \addr Department of Statistics and Operations Research\\
	   University of North Carolina at Chapel Hill\\
	   Chapel Hill, NC 27599}\footnotetext{\noindent$^*$JDW is corresponding author. JDW and JP contributed equally to the writing of this paper.}

\editor{Edoardo Airoldi}

\maketitle

\begin{abstract}
	Multilayer networks are a useful way to capture and model multiple, binary or weighted relationships among a fixed group of objects.  While community detection has proven to be a useful exploratory technique for the analysis of single-layer networks, the development of community detection methods for multilayer networks is still in its infancy. 
			We propose and investigate a procedure, called Multilayer Extraction, that identifies densely connected 
			vertex-layer sets in multilayer networks. 
			Multilayer Extraction makes use of a significance based score that quantifies the connectivity of an observed vertex-layer set 
			through comparison with a fixed degree random graph model. 
			Multilayer Extraction directly handles networks with heterogeneous layers where community structure may be different from layer to layer.
			The procedure can capture overlapping communities, as well as background vertex-layer pairs that do not
			belong to any community. 
			We establish consistency of the vertex-layer set optimizer of our proposed multilayer score under the multilayer stochastic block model.
			We investigate the performance of Multilayer Extraction on three applications and a test bed of simulations. 
			Our theoretical and numerical evaluations suggest that Multilayer Extraction is an effective exploratory tool for analyzing complex multilayer networks. Publicly available code is available at \url{https://github.com/jdwilson4/MultilayerExtraction}.
\end{abstract}
\vskip 1pc


\newpage

\section{Introduction}\label{sec:intro}

Networks are widely used to represent and analyze the relational structure among interacting units of a complex system. 
In the simplest case, a network model is an unweighted, undirected graph $G = (V,E)$, 
where $V$ is a set of vertices that represent the units, or {\it actors}, 
of the modeled system, and $E$ is an edge set containing all pairs of vertices $\{u,v\}$ such that actors
$u$ and $v$ share a physical or 
functional relationship. Networks have been successfully applied in {a wide array of fields, including} the social sciences to study social relationships among individuals \citep{wasserman1994advances},
biology to study interactions among genes and proteins \citep{bader2003automated}, and
neuroscience to study the structure and function of the brain \citep{sporns2011networks}. 

In many cases, the vertices of a network can be divided into groups (often disjoint) with the property that 
there are many edges between vertices in the same group, but relatively few edges between vertices in different
groups. Vertex groups of this sort are commonly referred to as {\it communities}. 
The unsupervised search for communities in a network is known as {\it community detection}.  Community structure has been used to identify functionally relevant groups in gene and protein interaction systems \citep{Lewis2010, parker2015network}, 
structural brain networks \citep{Bassett2011}, and social networks \citep{Onnela2011, Greene2010}. 
As communities are often associated with important structural characteristics of a complex system,
community detection is a common first step in the understanding and analysis of networks. The search for communities that optimize a given quantitative performance criterion is typically an NP-hard problem, so in most cases one must rely on approximate algorithms to identify community structure. 

The focus of this paper is community detection in {\it multilayer} networks.
Formally, an $(m,n)$-multilayer network is a collection 
$\mathbf{G}(m,n) = ( G_1, \ldots, G_m )$ of $m$ simple graphs $G_{\ell} = ([n], E_\ell)$
having common vertex set $[n] = \{1,\ldots, n\}$, where the edge sets $E_\ell$ may vary from
layer to layer.   
The graph $G_{\ell}$ will be referred to as the
$\ell$th {\em layer} of the network. 
We assume that the vertices of the multilayer network are registered, in the sense that 
a fixed vertex $u \in [n]$ refers to the same actor across all layers.  
Thus the graph $G_{\ell}$ reflects the relationships between 
identified actors $1,\ldots,n$ in circumstance $\ell$.
There are no edges between vertices in different layers, and the layers are
regarded as unordered so that the indices $\ell \in [m]$ do not 
reflect an underlying spatial or temporal order among the layers. 

In general, the actors of a multilayer network may not exhibit the same community structure across all layers. 
For example in social networks, a group of individuals may be well connected via friendships on Facebook; however, this common group of actors will likely, for example, not work at the same company.
In realistic situations such as these, a given vertex community will only be present in a subset 
of the layers, and different communities may be present in different subsets of layers. We refer to such multilayer systems as heterogeneous as each layer may exhibit noticeably different community structure.
Complex and differential relationships between actors will be reflected in the heterogenous behavior
of different layers of a multilayer network.  In spite of this heterogeneity, many existing community detection methods for multilayer networks typically assume that the community structure is the same across all, or a substantial fraction of, the layers \citep{berlingerio2011finding, rocklin2013clustering, barigozzi2011identifying, berlingerio2013abacus, holland1983stochastic, paul2015community}.

We develop and investigate a multilayer community detection method called Multilayer 
Extraction, which efficiently handles multilayer networks with heterogeneous layers. Theoretical and numerical evaluation of our method reveals that Multilayer Extraction is an effective exploratory tool for analyzing complex multilayer networks. Our contributions to the current literature of statistical analysis of multilayer networks are threefold

\begin{enumerate} 
\item We develop a testing-based algorithm for identifying densely connected vertex-layer sets $(B,L)$, where $B \subseteq [n]$ is a set of vertices and $L \subseteq [m]$ is a set of layers such that the vertices in $B$ are densely connected across the layers in $L$. The strength of the connections in $(B,L)$ is measured by a local modularity score derived from a null random network model that is based on the degree sequence of the observed multilayer network. Identified communities can have overlapping vertex or layer sets, and some vertex-layer pairs may not belong to any community. Vertex-layer pairs that are not assigned to any community are interpreted as background as they are not strongly connected to any other.
Overlap and background are common features of real networks that
can have deleterious effects
on partition based methods \citep{lancichinetti2011finding, wilson2013measuring, wilson2014testing}. The Multilayer Extraction procedure directly addresses community heterogeneity in multilayer networks.

\item We assess the consistency of the global optimizer of the aforementioned local modularity score under a generative model for multilayer networks with communities. The generative model studied is a multilayer generalization of the stochastic 2 block model from \cite{snijders1997estimation, wang1987stochastic} for multilayer networks, which we call the Multilayer Stochastic Block Model (MSBM). The MSBM is a generative model that characterizes preferential (or a-preferential) behavior of pre-specified vertex-layer communities in a multilayer network, via specification of inter- and intra-community probabilities. We are able to show that under the MSBM, the number of mis-clustered vertices and layers from the vertex-layer community that maximizes our proposed significance score vanishes to zero, with high probability, as the number of vertices tends to infinity. There has been considerable work in the area of consistency analysis for single-layer networks \cite[e.g.][]{zhao2012consistency}; however, to the best of our knowledge, we are the first to address the joint optimality properties for \emph{both} vertices and layers in multilayer networks. Furthermore, we provide complete and explicit expressions of all error bounds in the proof, as we anticipate future analyses where the number of layers is allowed to grow with the size of the network. Our proof involves a novel inductive argument, which, to our knowledge, has not been employed elsewhere.

\item We apply Multilayer Extraction to three diverse and important multilayer network types, including a multilayer social network, a transportation network, and a collaboration network. We compare and contrast Multilayer Extraction with contemporary methods, and highlight the advantages of our approach over single layer and aggregate alternatives. Our findings reveal important insights about these three complex relational systems.

\end{enumerate}

\subsection{Related Work}\label{sec:relatedwork}
Multilayer network models have been applied to a variety of problems, including 
modeling and analysis of air transportation routes \citep{cardillo2013emergence}, 
studying individuals with multiple sociometric relations 
\citep{fienberg1980analyzing, fienberg1985statistical}, and analyzing relationships between social 
interactions and economic exchange \citep{ferriani2013social}. 
\citet{kivela2014multilayer} and \citet{boccaletti2014structure} provide two recent reviews of multilayer networks. We note that $\mathbf{G}(m,n)$ is also sometimes referred to as a {\it multiplex} network.

While there is a large and growing literature concerning community detection in standard, single-layer,
networks \citep{fortunato2010community, newman2004detecting, porter2009communities},
the development of community detection methods for multilayer networks is still 
relatively new.  One common approach to multilayer community detection is to project the multilayer network in some fashion onto a single-layer network and then identify communities in the single layer network \citep{berlingerio2011finding, rocklin2013clustering}. 
A second common approach to multilayer community detection is to apply a standard detection method to each layer of the observed network separately \citep{barigozzi2011identifying, berlingerio2013abacus}. The first approach fails to account for layer-specific community structure 
and may give an oversimplified or incomplete summary of the community structure of the multilayer network; 
the second approach does not enable one to leverage or identify common structure between layers. In addition to the methods above, there have also been several generalizations of single-layer methods to multilayer networks. For example, \citet{holland1983stochastic} and \citet{paul2015community} introduce multilayer generalizations of the stochastic block model from \citet{wang1987stochastic} and \citet{snijders1997estimation}. Notably, these generative models require the community structure to be the same across layers. 

\citet{peixoto2015inferring} considers a multilayer generalization of the stochastic block model for weighted networks that models hierarchical community structure as well as the degree distribution of an observed network. \citet{paul2016null} describe a class of null models for multilayer community detection based on the configuration and expected degree model. We utilize a similar model in our consideration. \citet{stanley2015clustering} considered the clustering of layers of multilayer networks based on recurring community structure throughout the network. \citet{peixoto2015inferring} and \citet{mucha2010community} generalized the notion of modularity to multilayer networks, and
\citet{de2014identifying} generalized the map equation, which measures the description length of a random walk on a partition of vertices, to multilayer networks. \citet{de2013mathematical} discuss a generalization of the multilayer method in \citet{mucha2010community} using tensor decompositions. Approximate optimization of either multilayer modularity or the map equation results in vertex-layer communities that form a partition of $[n] \times [m]$. By contrast, the communities identified by Multilayer Extraction need not form a partition of $[n] \times [m]$ as some vertices and or layers may be considered background and not significantly associated with any community.

\subsection{Overview of the Paper}

In the next section we describe the null multilayer random graph model and the scoring of vertex-layer sets. 
In Section \ref{sec:consistency} we present and prove theoretical results regarding the asymptotic consistency properties of our proposed score for multilayer networks. Section \ref{sec:algorithm} provides a detailed description of the proposed Multilayer Extraction procedure. 
We apply Multilayer Extraction to three real-world multilayer networks and compare and contrast its performance with existing community detection methods in Section \ref{sec:application}. 
In Section \ref{sec:simulations} we evaluate the performance of Multilayer Extraction on a test bed of simulated multilayer networks. We conclude the main paper with a discussion of future research directions in Section \ref{sec:discussion}. The Appendix is divided into three sections. In Appendix A, we prove supporting lemmas contributing to the results given in Section \ref{sec:consistency}. In Appendix B, we discuss competing methods to Multilayer Extraction. In Appendix C, we give the complete details of our simulation framework.

\section{Scoring a Vertex-Layer Group}\label{sec:model}

Seeking a vertex partition that optimizes, or approximately optimizes, an appropriate score function is a standard approach to single layer community detection.  
Prominent examples of score-based approaches include
modularity maximization \citep{newman2006modularity}, likelihood maximization for a
stochastic block model \citep{wang1987stochastic}, as well as minimization of the 
conductance of a partition \citep{chung1997spectral}.
Rather than scoring a partition of the available network, Multilayer Extraction makes use of a significance based score that is applicable to individual vertex-layer sets.  
Below, we describe the multilayer null model, and then the 
proposed score.

\subsection{The Null Model}\label{sec:null}

Our significance-based score for vertex-layer sets in multilayer networks relies on the comparison of an observed multilayer network with a null multilayer network model. Let $\mathbf{G}(m,n)$ be an observed $(m,n)$-multilayer network.  For each layer $\ell \in [m]$ and 
pair of vertices $u,v \in [n]$, let 
$$x_{\ell}(u,v) = \mathbb{I}( \{u,v\} \in E_{\ell} )$$
indicate the presence or absence of an edge between $u$ and $v$ in layer $\ell$ of $\mathbf{G}(m,n)$.
The {\em degree} of a vertex $u \in [n]$ in layer $\ell$, denoted by $d_{\ell}(u)$, is the number of edges 
incident on $u$ in $G_\ell$.  Formally,
\[
d_{\ell}(u) \ = \ \sum_{v \in [n]} x_{\ell}(u,v).
\] 
The {\em degree sequence} of layer $\ell$ is the vector $\mvd_{\ell} = ( d_{\ell}(1), \ldots, d_{\ell}(n) )$
of degrees in that layer; the degree sequence of $\mathbf{G}(m,n)$ 
is the list $\mvd = ( \mvd_{1}, \ldots, \mvd_{m} )$ containing the degree sequence of each layer in the network.

Let $\mathcal{G}(m,n)$ denote the family of all $(m,n)$-multilayer 
networks.  Given the degree sequence $\mvd$ of
the observed network $\mathbf{G}(m,n)$, we define a multilayer configuration model and an associated probability measure 
$\mathbb{P}_{\mvd}$ on $\mathcal{G}(m,n)$, as follows. In layer $G_1$, each node is given $d_1(u)$ half-stubs. Pairs of these edge stubs are then chosen uniformly at random, to form edges until all half-stubs are exhausted (disallowing self-loops and multiple edges). This process is done for every subsequent layer $G_2, \ldots, G_m$ independently, using the corresponding degree sequence from each layer.

In the multilayer network model described above, each layer is distributed according to the configuration model, first introduced by \cite{bollobas1979probabilistic} and \cite{bender1978asymptotic}. The probability of an edge between nodes $u$ and $v$ in layer $\ell$ depends only on the degree sequence ${\mvd}_\ell$ of the observed graph $G_\ell$. The distribution $\mathbb{P}_{\mvd}$ has
two complementary properties that make it useful for identifying communities in an observed multilayer network: 
(i) it preserves the degree structure of the 
observed network; 
and (ii) subject to this restriction, edges are assigned at random, without regard to
the higher order connectivity structure of the network. Because of these characteristics, the configuration model has long been taken as the appropriate null model against which to judge the quality of a proposed community partition. 

The configuration model is the null model which motivates the modularity score of a partition in a network \citep{newman2004detecting, newman2006modularity}. Consider a single-layer observed network $\mathbf{G}(n)$ with $n$ nodes and degree sequence $\mvd$. For fixed $K>0$, let $c_u \in[K]$ be the community assignment of node $u$. The modularity score of the partition associated with the assignment $c_1, \ldots, c_n$ is defined as
\begin{equation}\label{full-modularity}
M(c_1,\ldots,c_n;\;\mathbf{G}(n)) := \frac{1}{2|E|}\sum_{i\in[K]}\sum_{u < v \in [n]}\left(x(u,v) - \frac{d(u)d(v)}{\sum_{w \in [n]}d(w)}\right)\mathbb{I}(c_u = c_v = i).
\end{equation}
Above, the ratio $\frac{d(u)d(v)}{\sum_{w \in [n]}d(w)}$ is the approximate expected number of edges between $u$ and $v$ under the configuration model. If the partition $c_1, \ldots, c_n$ represents communities with a large observed intra-edge count relative to what is expected under the configuration model, it receives a high modularity score. The identification of the communities that (approximately) maximize the modularity of a partition is among the most common techniques for community detection in networks.

\subsection{Multilayer Extraction Score}\label{subsec:score}

Rather than scoring a partition, the Multilayer Extraction method scores individual vertex-layer sets. We define a multilayer node score that is based on the single-layer modularity score \eqref{full-modularity} and amenable to iterative maximization. We first define a local \emph{set} modularity for a collection of vertices $B \subseteq [n]$ in the layer $\ell \in [m]$:
\begin{equation}\label{set-modularity}
Q_\ell(B) := \frac{1}{n{|B|\choose 2}^{1/2}}\sum_{u,v\in B:u<v}\left(x_\ell(u,v) - \frac{d_\ell(u)d_\ell(v)}{\sum_{w \in [n]} d_{\ell}(w)}\right).
\end{equation}
The scaling term in the equation above is related to the total number of vertices in the network and the total number of possible edges between the vertices in $B$. This score is one version of the various set-modularities considered in \cite{fasino2016modularity}, and is reminiscent of the \emph{local} modularity score introduced in \cite{clauset2005finding}.

Our Multilayer Extraction procedure seeks communities that are \emph{assortative} across layers, in the sense that $Q_\ell(B)$ is large and positive for each $\ell\in L$. In light of this, we define the \emph{multilayer set score} as
\begin{equation}\label{score}
H(B, L) := \frac{1}{|L|}\left(\sum_{\ell\in L}Q_\ell(B)_+\right)^2,
\end{equation}
where $Q_+$ denotes the positive part of $Q$. Generally speaking, the score acts as a yardstick with which one can measure the connection strength of a vertex-layer set. Large values of the score signify densely connected communities.

We note that the multilayer score $H(B,L)$ is reminiscent of a chi-squared test-statistic computed from $|L|$ samples. That is, under appropriate regularity assumptions on $Q_{\ell}(B)$, the score in (\ref{score}) will be approximately chi-squared with one degree of freedom.

\section{Consistency Analysis}\label{sec:consistency}

In this section we give theoretical results establishing the asymptotic consistency of the Multilayer Extraction score introduced in Section \ref{subsec:score}. Explicitly, we show that the global optimizer of the score has bounded misclassification error under a 2-community multilayer stochastic block model. 

We note that our results are particular to the performance of the score, rather than the performance of the algorithm to maximize the score. This approach to consistency analysis is common in the community detection literature, particularly for modularity-based methods, as maximizing modularity is NP-hard \citep{brandes2006maximizing}. For instance, \cite{bickel2009nonparametric} showed that under the standard (single-layer) stochastic block model, the global optimizer of the modularity partition score converges (in classification error) to the partition given by the model.  \cite{zhao2012consistency} extended this result to a generalized class of partition scores, which includes both modularity and the SBM likelihood, and established conditions under which global optimizers of those scores are consistent under the \emph{degree-corrected} stochastic block model (an extension of the standard version, see \citealp*{coja2009finding}). In \cite{zhao2011community}, a global-optimizer consistency result for a \emph{local} set score, one \emph{not} based on the modularity metric, was established in similar fashion. Many other analyses assess the theoretical ``detectability" of the community structure in the stochastic block model (e.g.\ \citealp*{decelle2011asymptotic, mossel2012stochastic, nadakuditi2012graph}).

As with the theorems we present in this section, most results from the aforementioned publications establish, or help to establish, asymptotic ``with high probability" statements, and apply only to the \emph{global} maximizers of the scores under study. In other words, they do not directly regard in-practice approaches to local maximization of the scores. Regardless, the demonstration of global optimizer consistency for a given score provides the measure of comfort that any local optimization method has a sensible target. The following analyses will serve our method in accordance with this intuition.

\subsection{The Multilayer Stochastic Block Model}

We assess the consistency of the Multilayer Extraction score under the multilayer stochastic block model (MSBM) with two vertex communities, defined as a probability distribution $\mathbb{P}_{m,n} = \mathbb{P}_{m,n}(\cdot | \bm{P}, \pi_1, \pi_2)$ on the family of undirected multilayer networks with $m$ layers, $n$ vertices and 2 communities. The distribution is fully characterized by
(i) containment probabilities $\pi_1, \pi_2 > 0$, which satisfy $\pi_1 + \pi_2 = 1$,
and (ii) a sequence of symmetric $2 \times 2$ matrices $\bm{P} = \{P_1, \ldots, P_m\}$ where $P_\ell =  \{ P_\ell(i,j) \}$ with entries $P_{\ell}(i,j) \in (0,1)$.  Under the distribution $\mathbb{P}_{m,n}$, a random multilayer network $\widehat{\bm{G}}(m,n)$ is generated using two simple steps:

\begin{enumerate}
\item A subset of $\lceil{\pi_1 n}\rceil$ vertices are placed in community 1, and remaining vertices are placed in community 2. Each vertex $u$ in community $j$ is assigned a community label $c_u = j$. 
\item An edge is placed between nodes $u, v \in [n]$ in layer $\ell \in [m]$ with probability $P_{\ell}(c_u,c_v)$, independently from pair to pair and across layers, and no self-loops are allowed.
\end{enumerate}

For a fixed $n$ and $m$, the community labels ${\mathbf c}_n = (c_1, \ldots, c_n)$ are chosen once and only once, and the community labels are the same across each layer of $\widehat{\bm{G}}(m,n)$. On the other hand, the inner and intra community connection probabilities (and hence the assortativity) can be different from layer to layer, introducing heterogeneity among the layers. Note that when $m = 1$, the MSBM reduces to the (single-layer) stochastic block model from \citet{wang1987stochastic}.

\subsection{Consistency of the Score} \label{subsec:score-consistency}
We evaluate the consistency of the Multilayer Extraction score under the MSBM described above. Our first result addresses the vertex set maximizer of the score given a fixed layer set $L\subseteq[m]$. Our second result (Theorem \ref{layer-set-consistency} in Section \ref{layer-set-theory}) leverages the former to analyze the global maximizer of the score across layers and vertex sets. Explicitly, consider a multilayer network $\widehat{\bm{G}}(m,n)$ with distribution under the multilayer stochastic block model $\mathbb{P}_{m,n} = \mathbb{P}_{m,n}(\bm{P}, \pi_1, \pi_2)$. For a fixed vertex set $B \subseteq [n]$ and layer set $L \subseteq [m]$, define the random score by
\[
\widehat{H}(B, L):=\frac{1}{|L|}\left(\sum_{\ell\in L}\widehat Q_\ell(B)_+\right)^2,
\]
where $\widehat Q_\ell(B)$ is the set-modularity of $B$ in layer $\ell$ under $\mathbb{P}_{m,n}$. Our main results address the behavior of $\widehat H(B,L)$ under various assuptions on the parameters of the MSBM.

Toward the first  result, for a fixed layer set $L\subseteq[m]$, let $\Bopt$ denote the node set that maximizes $\widehat H(B, L)$ (if more than one set does, any may be chosen arbitrarily). To define the notion of a ``misclassified" node, for any two sets $B_1,B_2\subseteq[n]$ let $d_h(B_1, B_2)$ denote the Hamming distance (rigorously defined as the cardinality of the symmetric difference between $B_1$ and $B_2$). We then define the number of misclassified nodes by a set $B$ by 
\[
\mathtt{Error}(B) := \min\{d_h(B, C_1),\; d_h(B, C_2)\}.
\]
Note that this definition accounts for arbitrary labeling of the two communities. As the nodes and community assignments are registered across layers, neither $d_h$ nor $\mathtt{Error}$ depend on the choice of $L$. Before stating the main theorem, we define a few quantities that will be used throughout its statement and proof:

\begin{definition}\label{parameter-defs}
Let ``$\det$'' denote matrix determinant. For a fixed layer set $L\subseteq[m]$, define
\begin{equation}
\delta_\ell := \det P_\ell\;\;\;\;\;\delta(L) := \min_{\ell\in[L]}\delta_\ell\;\;\;\;\;\pi:=(\pi_1,\pi_2)^t\;\;\;\;\;\kappa_\ell := \pi^TP_\ell\pi\;\;\;\;\;\kappa(L) := \min_{\ell\in[L]}\kappa_\ell
\end{equation}
\end{definition}

We now state the fixed-layer-set consistency result:

\begin{theorem}\label{thm:largegraph}
Fix $m$ and let $\{\widehat{\bm{G}}({m,n})\}_{n > 1}$ be a sequence of multilayer stochastic 2 block models where $\widehat{\bm{G}}({m,n})$ is a random graph with $m$ layers and $n$ nodes generated under $\mathbb{P}_{m,n}(\cdot | \bm{P}, \pi_1, \pi_2)$. Assume $\pi_1\leq\pi_2$, and that $\pi_1$, $\pi_2$, and $\mathbb{P}$ do not change with $n$. Fix a layer set $L \subseteq [m]$. If $\delta(L) > 0$ then there exist constants $A,\eta>0$ depending on $\pi_1$ and $\delta(L)$ such that for all fixed $\eps\in(0,\eta)$,
\begin{equation}\label{eq:largesampleconsistency}
\mathbb{P}_{m,n}\left(\mathtt{Error}\left(\Bopt\right) < An^\eps\log n\right) \geq 1 - \exp\left\{-\frac{\kappa(L)^2\eps}{32}n^\eps(\log n)^{2-\eps} + \log 4|L|\right\}
\end{equation}
for large enough $n$.
\end{theorem}
Note that the right-hand-side of \eqref{eq:largesampleconsistency} converges to 1 for all $\eps\in(0,1)$, regardless of $\eta$. Furthermore, if $\eps\in(0,\eta)$, we have $n^\eps<n^{\eps'}$ for all $\eps'\geq\eta$. Therefore, a corollary of Theorem \ref{thm:largegraph} is that for any $\eps\in(0,1)$, 
\[
\mathbb{P}_{m,n}\left(\mathtt{Error}\left(\Bopt\right) < n^\eps\log n\right)\rightarrow1\;\text{ as }\;n\rightarrow\infty.
\]
The above statement is perhaps a more illustrative version of the Theorem \ref{thm:largegraph}, and shows that the constants $A$ and $\eta$ play a role only in bounding the convergence rate of the probability.

The proof of Theorem \ref{thm:largegraph} is given in Section \ref{thm2-proof-sec}. We note that the assumption that $\pi_1\leq\pi_2$ is made without loss of generality, since the community labels are arbitrary. When $m = 1$, Theorem \ref{thm:largegraph} implies asymptotic $n\rightarrow\infty$ consistency in the (single-layer) stochastic block model. In this case, the condition that $\delta_\ell = {P}_{\ell}(1,1) {P}_{\ell}(2,2) - {P}_{\ell}(1,2)^2>0$ is a natural requirement on the inner community edge density of a block model. This condition appears in a variety of consistency analyses, including the evaluation of modularity \citep{zhao2012consistency}. When $m > 1$, Theorem \ref{thm:largegraph} implies the vertex set that maximizes $H(B,L)$ will have asymptotically vanishing error with high probability, given that $L$ is a fixed layer set with \emph{all} layers satisfying $\delta_\ell > 0$. 

\subsubsection{Consistency of the joint optimizer}\label{layer-set-theory}
Theorem \ref{thm:largegraph} does not address the \emph{joint} optimizer of the score across all vertex-layer pairs. First, we point out that for a fixed $B\subseteq[n]$, the limiting behavior of the score $\widehat H(B, L)$ depends on $L\subseteq[m]$ through the layer-wise determinants $\{\delta_\ell:\ell\in[n]\}$ and the scaling constant $\frac{1}{|L|}$ inherent to $H(B,L)$, as defined in equation \eqref{score}. Let $\gamma:\mathbb{N}^+\mapsto\mathbb{R}^+$ be a non-decreasing function of $|L|$. Define
\begin{equation}\label{score2}
H_\gamma(B,L) := \frac{1}{\gamma(|L|)}\left(\sum_{\ell \in L}Q_\ell(B)_+\right)^2.
\end{equation}
and let $\widehat H_\gamma(B,L)$ be the corresponding random version of this score under an MSBM. We analyze the joint node-set optimizer of $H_\gamma$ under some representative choices of $\gamma$, an analysis which will ultimately motivate the choice $\gamma(|L|) = |L|$.

We first provide an illustrative example. Consider a MSBM with $m>1$ layers having the following structure: the first layer has positive determinant, and all $m-1$ remaining layers have determinant equal to 0. Note that $\delta_1>0$ implies that the first layer has ground-truth assortative community structure, and that $\delta_\ell = 0$ for all $\ell > 1$ implies that the remaining layers are (independent) Erdos-Renyi random graphs. In this case, the desired global optimizer of $H_\gamma(B, L)$ is community 1 (or 2) and the first layer. However, setting $\gamma(|L|) \equiv 1$ (effectively ignoring the scaling of $H$) will ensure that, in fact, the \emph{entire} layer set is optimal, since $Q_\ell(B)_+\geq0$ by definition. It follows that setting $\gamma(|L|)$ to increase (strictly) in $|L|$, which introduces a penalty on the size of the layer set, is desirable.

For a fixed scaling function $\gamma$, define the global joint optimizer of $\widehat H(B, L)$ by
\begin{equation}\label{jopt}
\Jopt := \underset{2^{[n]}\times 2^{[m]}}{\argmax}\; \widehat H_\gamma(B, L).
\end{equation}
Note that $\Jopt$ is random, and may contain multiple elements of $2^{[m]}\times 2^{[n]}$. The next theorem addresses the behavior of $\Jopt$ under the MSBM for various choices of $\gamma(|L|)$, and shows that setting $\gamma(|L|) = |L|$ is desirable for consistency.

\begin{theorem}\label{layer-set-consistency}
Fix $m$ and let $\{\widehat{\bm{G}}({m,n})\}_{n > 1}$ be a sequence of multilayer stochastic 2 block models where $\widehat{\bm{G}}({m,n})$ is a random graph with $m$ layers and $n$ nodes generated under $\mathbb{P}_{m,n}(\cdot | \bm{P}, \pi_1, \pi_2)$. Assume $\pi_1\leq\pi_2$, and that $\pi_1$, $\pi_2$, and $\mathbb{P}$ do not change with $n$. Fix $0=\delta^{(0)}<\delta^{(1)}<1$. Suppose the layer set $[m]$ is split according to $[m] = \cup_{i = 0, 1}L_i$, where $\delta_\ell = \delta^{(i)}$ for all $\ell\in L_i$. Then for any $\eps>0$, the following hold:
\begin{enumerate}[(a)]
\item Let $\widehat L^+ := \{\ell: \widehat Q_\ell(\widehat B_{opt}^{(n)})>0\}$. If $\gamma(|L|)\equiv 1$, then for all $n>1$, $\widehat L_{opt}^{(n)} = \widehat L^+$, and 
\[
\mathbb{P}_{m,n}\left(\mathtt{Error}\left(\widehat B_{opt}^{(n)}\right) < n^\eps\log n\right)\rightarrow1\;\text{ as }\;n\rightarrow\infty.
\]
\item If $\gamma(|L|) = |L|$,
\[
\mathbb{P}_{m,n}\left(\widehat L_{opt}^{(n)} = L_1,\;\mathtt{Error}\left(\widehat B_{opt}^{(n)}\right) < n^\eps\log n\right)\rightarrow1\;\text{ as }\;n\rightarrow\infty.
\]
\item If $\gamma(|L|) = |L|^2$,
\[
\mathbb{P}_{m,n}\left(\widehat L_{opt}^{(n)} \subseteq 2^{L_1},\;\mathtt{Error}\left(\widehat B_{opt}^{(n)}\right) < n^\eps\log n\right)\rightarrow1\;\text{ as }\;n\rightarrow\infty.
\]
\end{enumerate}
\end{theorem}

The proof of Theorem \ref{layer-set-consistency} is given in Section \ref{layer-set-proof}. Part (a) implies that setting $\gamma(|L|)\equiv1$ ensures that the optimal layer set will be, simply, all layers with positive modularity, thereby making this an undesirable choice for the function $\gamma$. Part (c) says that if $\gamma(|L|)= |L|^2$, the layer set with the highest \emph{average} layer-wise modularity will be optimal (with high probability as $n\rightarrow\infty$), which means that all subsets of $L_1$ are asymptotically equivalent with respect to $\widehat H(B,L)$ (with high probability). By part (b), if $\gamma(|L|) = |L|$, then $L_1$ is the unique asymptotic maximizer of the population score (with high probability). Therefore, $\gamma(|L|) = |L|$ is the most desirable choice of scaling.

\subsection{Discussion of theoretical results}\label{ss:theory-discussion}
As mentioned previously, the theoretical results in Section \ref{subsec:score-consistency} regard the \emph{global} optimizer of the score. The results have the following in-practice implication: if one simulates from an MSBM satisfying the assumptions of the theorem, and subsequently finds the \emph{global} optimizer of the score, the classification error of the optimizer will be vanishingly small with increasingly high probability (as $n\rightarrow\infty$). To illustrate this point, consider an MSBM with $m>1$ layers, each having community structure $\pi_1 = 0.4$ and for $r\in(0, .95)$, $P(1, 1) = P(2, 2) = 0.05 + r$ and $P(1, 2) = 0.05$. Under this model, $\delta_\ell = r(r + 0.1) > 0$ for all $\ell\in[m]$, and the assumptions of Theorem \ref{layer-set-consistency} are satisfied. Therefore, if we were to simulate this MSBM and find the global optimizer of $H$, with high probability (increasingly, $n\rightarrow\infty$) we would recover (1) the correct layer set $[m]$, and (2) the optimal node set will have small classification error (vanishingly, $n\rightarrow\infty$).

Of course, in practice it is computationally infeasible to find the global optimizer, so we employ the method laid out in Section \ref{sec:algorithm} to find local maxima. We find through simulation (see Section \ref{sec:simulations}) that our method achieves extremely low error rates for relatively small networks, including on the  MSBM described above, for many values of $r$. These results reflect the intuition that theoretical results for global optimizers should have practical implications when local optimization methods are sufficiently effective.

A limitation of our theoretical results is that they assume an MSBM with only two communities. Furthermore, the model does not allow for any notion of overlap between communities. Nevertheless, in the setting we consider we demonstrate that there is a concise condition on the parameters of the model that guarantees consistency, namely, that $\delta(L) > 0$ for a target layer set $L$. We expect a similar condition to exist in more complicated settings. In particular, at the outset of Section \ref{thm2-proof-sec}, we sketch how the $\delta(L)>0$ condition relates to the maxima of the \emph{population} version of the score, which roughly speaking is the deterministic, limiting form of the score under the MSBM. We expect that, in more complex settings with (say) more than two communities or overlapping nodes, similar conditions on the model parameters would guarantee that the population version of the score be maximized at the correct community/layer partitions, which (as we show) would entail that the classification error of the global maximizer converges in probability to zero. Though the proofs in such settings would undoubtedly be more complicated, the analyses in this paper should serve as a theoretical foundation. Furthermore, we have empirically analyzed those settings via simulation in Section \ref{sec:msbm} (for more than two communities) and Appendix \ref{sec:appendix-sims} (for overlapping communities).

\section{Proofs}

In this section we prove the theoretical results given in Section \ref{sec:consistency}. The majority of the section is devoted to a detailed proof of Theorem \ref{thm:largegraph} and supporting lemmas. This is followed by the proof of Theorem \ref{layer-set-consistency}, of which we give only a sketch, as many of the results and techniques contributing to the proof of Theorem \ref{thm:largegraph} can be re-used.




\subsection{Proof of Theorem \ref{thm:largegraph}, and Supporting Lemmas}\label{thm2-proof-sec}

We prove Theorem \ref{thm:largegraph} via a number of supporting lemmas. We begin with some notation:

\begin{definition}\label{rho}
For a fixed vertex set $B\subseteq[n]$ define
\begin{equation}
	\rho_n(B) = \dfrac{|B \cap C_{1,n}|}{|B|}, \hskip 1pc s_n(B) = \dfrac{|B|}{n}, \hskip 1pc  \Bv_n(B) := (\rho_n(B), 1 - \rho_n(B)).
\end{equation}
We will at times suppress dependence on $n$ and $B$ in the above expressions.
\end{definition}
\begin{definition}\label{def:pop-score}
Define the \textbf{population} normalized modularity of a set $B$ in layer $\ell$ by
\begin{equation}\label{pop-mod}
\cQ_\ell(B) := \frac{s_n(B)}{\sqrt{2}}\left(\Bv_n(B)^tP_\ell \Bv_n(B) - \dfrac{(v_n(B)^t P_\ell \pi)^2}{\kappa_\ell}\right).
\end{equation}
Define the \textbf{population} score function $\cH(\cdot, L):2^{[n]}\mapsto\mathbb{R}$ by 
\begin{equation}\label{pop-score}
\cH(B, L) = |L|^{-1}\left(\sum_{\ell\in[L]}\cQ_\ell(B)\right)^2.
\end{equation}
\end{definition}

Throughout the results in this section, we assume that $L\subseteq[m]$ is a fixed layer set (as in the statement of Theorem \ref{thm:largegraph}). We will therefore, at times, suppress the dependence on $L$ from $\delta(L)$ and $\kappa(L)$ (from Definition \ref{parameter-defs}).

\subsubsection{Sketch of the Proof of Theorem \ref{thm:largegraph}}
The proof of Theorem \ref{thm:largegraph} is involved and broken into many lemmas. In this section, we give a rough sketch of the argument, as follows. The lemmas in this section establish that:
\begin{enumerate}
\item $C_{1,n}$ maximizes the population score $H_\ast(\cdot, L)$ (Lemmas \ref{pop-score-rep} and \ref{prop:optimizer}).
\item For large enough sets $B\subseteq[n]$, the random score $\widehat H(B, L)$ is bounded in probability around the population score $H_\ast(B, L)$ (Lemmas \ref{prop:sup-score-conc-ineq} and \ref{conc-cor}).
\item \textbf{Inductive Step}: For fixed $k>1$, assume that $d_h(\Bopt, C_{1,n})/n = O_p(b_{n,k})$, where larger $k$ makes $b_{n,k}$ of smaller order. Then, based on concentration properties of the score, in fact $d_h(\Bopt, C_{1,n})/n = O_p(b_{n,k+1})$ (Lemma \ref{lem:refinement}).
\item There exists a constant $\eta$ such that for any $\eps\in (0, \eta)$, $d_h(\Bopt, C_{1,n})/n = O_p(n^\eps\log n)$ (Theorem \ref{thm:largegraph}). This result is shown using the Inductive Step.
\end{enumerate}

\subsubsection{Supporting lemmas for the Proof of Theorem \ref{thm:largegraph}}\label{thm2-lemmas}

\begin{lemma}\label{pop-score-rep}
Define $\phi(L) := (|L|^{-1}\sum_\ell\frac{\det P_\ell}{2\kappa_\ell})^2$. Then:
\begin{enumerate}
\item For any $B\subseteq[n]$,  $q_\ell(B) = \frac{s_n(B)}{\sqrt{2}}(\pi_1 - \rho_n(B))^2\cdot \frac{\det P_\ell}{2\kappa_\ell}$, and therefore 
\[
\cH(B, L) = |L|\phi(L)\frac{s_n(B)^2}{2}(\pi_1 - \rho_n(B))^4.
\]
\item $\delta(L)^2\leq\phi(L)\leq\frac{1}{\pi_1^2\delta(L)^2}$ and therefore $\cH(C_{1,n}, L) \geq |L|\frac{\pi_1^2}{2}(1 - \pi_1^4)\delta(L)^2$.
\end{enumerate}
\end{lemma}

\begin{lemma}\label{prop:optimizer}
Fix any $n>1$. Define 
\[
\cR(t) := \begin{cases}
\big\{B\subseteq[n]:\max\{|s(B)-\pi_1|,\;1-\rho(B)\}\leq t\big\},&\pi_1<\pi_2\\
\big\{B\subseteq[n]:\max\{|s(B)-\pi_1|,\;\rho(B)\}\leq 1-\rho(B) \leq t\big\},& \pi_1 = \pi_2
\end{cases}
\]
Then there exists a constant $a>0$ depending just on $\pi_1$ such that for sufficiently small $t$, $B\notin\cR(t)$ implies $H_\ast(B, L)<H_\ast(C_{1,n}, L) - a|L|\phi(L)t$. 
\end{lemma}

The proofs of Lemmas \ref{pop-score-rep}-\ref{prop:optimizer} are given in Appendix \ref{sec:AppA}. We now give a general concentration inequality for $\widehat H(B, L)$, which shows that for sufficiently large sets $B\subseteq[n]$, $\widehat H(B,L)$ is close to the population score $H_\ast(B, L)$ with high probability. This result is used in the proof of Lemma \ref{conc-cor}, and its proof is given in Appendix \ref{sec:AppA}. We first give the following definition:

\begin{definition} For fixed $\eps>0$ and $n>1$, define $\cB_n(\eps) := \{B\subseteq[n]:|B|\geq n\eps\}$.
\end{definition}

\begin{lemma}
	\label{prop:sup-score-conc-ineq}
Fix $\eps\in(0, 1)$. Let $\kappa$ be as in Definition \ref{parameter-defs}. For each $n>1$ suppose a collection of node sets $\cB_n$ is contained in $\cB_n(\eps)$. Then for large enough $n$,
\[
		\mathbb{P}_n\left(\sup_{\mathcal{B}_n}\left(\Big|\widehat{H}(B,L) - \cH(B, L)\Big|\right) > \dfrac{4|L|t}{n^2} + \frac{52|L|}{\kappa n}\right) \leq 4|L||\mathcal{B}_n| \exp\left(-\kappa^2 \dfrac{\eps t^2}{16n^2}\right)
\]
for all $t>0$.
	\end{lemma}

We now define new notation that will serve the remaining lemmas:
\begin{definition}\label{gamma-n}
Let $\gamma_n := \log n/n$, and for any integer $k>0$, define $b_{n,k} := \gamma_n^{1-\halfk}$.
\end{definition}

\begin{definition}\label{N-ast}
For any $r\in[0,1]$ and $C\subseteq[n]$, define the $r$-neighborhood of $C$ by $N(C, r) := \{B\subseteq[n]:d_h(B, C)/n\leq r\}$. For all $n>1$, any constant $A>0$, and fixed $k>1$, define
\[
\SmallNa := \begin{cases}
N\left(C_1, A\cdot b_{n,k-1}\right)\cup N\left(C_2, A\cdot b_{n,k-1}\right),&k>1\\
\cB_n(A),&k=1
\end{cases}
\]
\end{definition}

Lemma \ref{conc-cor}, stated below, is a concentration inequality for the random variable $\widehat H(B, L)$ on particular neighborhoods of $C_1$:

\begin{lemma}\label{conc-cor}
Fix $\eps\in(0, \pi_1)$ and any constant $A>0$. For $k>1$ satisfying $1 / 2^{k-1}<\eps$, we have for sufficiently large $n$ that
\begin{equation}
\sup_{B\in \SmallNa}\Big|\widehat{H}(B, L) - \cH(B, L)\Big| \leq 5|L|b_{n,k}.
\end{equation}
with probability greater than $1 - 2\exp\{-\frac{\kappa^2\eps}{32}n\gamma_n^{1-\eps}\log(n) + \log4|L|\}$. The conclusion holds with $k=1$ if $A = \eps$.
\end{lemma}

The proof of Lemma \ref{conc-cor} is given in Appendix \ref{sec:AppA}.  We now state and prove the key lemma used to drive the induction step in the proof of Theorem \ref{thm:largegraph}:

\begin{lemma}
	\label{lem:refinement}
	Fix $\eps\in(0, \pi_1)$ and an integer $k>1$ satisfying $1/2^{k-1}<\eps$. Suppose there exist constants $A,b>0$ such that for large enough $n$,
\[
\mathbb{P}_n\left(\widehat B_{opt}(n)\in\SmallNa\right) \geq 1 - b\exp\left\{-\frac{\kappa^2\eps}{32}n\gamma_n^{1-\eps}\log n + \log4|L|\right\} := 1-b\beta_n(\eps)
\]
Then there exists a constant $A'>0$ depending only on $\pi_1$ and $\delta$ such that for large enough $n$, $\mathbb{P}_n\left(\widehat B_{opt}(n)\in\SmallNap\right)\geq 1 - (4+b)\beta_n(\eps)$. The conclusion holds for $k = 1$ if $A = \eps$.
\end{lemma}




\begin{proof} Assume $\pi_1<\pi_2$; the following argument may be easily adapted to the case where $\pi_1 = \pi_2$, which we explain at the end. Recall $b_{n,k}$ from Definition \ref{gamma-n}. For $c>0$, define
\[
\cR_{n,k}(c) := \big\{B\subset[n]:\max\{|s(B) - \pi_1|,\; 1 - \rho(B)\}\leq c\cdot b_{n,k}\big\},
\]
Note that sets $B\in\cR_{n,k}(c)$ have bounded Hamming distance from $C_{1,n}$, as shown by the following derivation. Writing $s = s(B)$ and $\rho = \rho(B)$, for all $B\in\cR_{n,k}(c)$ we have
\begin{align}
n^{-1}|d_h(B, C_{1,n})| & = n^{-1}\big(|B\setminus C_{1,n}| + |C_{1,n}\setminus B|\big)\nonumber\\[1em]
 &= n^{-1}\big(|B| - |B\cap C_{1,n}| + |C_{1,n}| - |B\cap C_{1,n}|\big)\nonumber\\[1em]
& = s + \pi_1 - 2\rho s \nonumber\\[1em]
&\leq s + \left(s + c\cdot b_{n,k}\right) - 2\left(1 - c\cdot b_{n,k}\right)s\nonumber\\[1em]
& = c\cdot b_{n,k} +2sc\cdot b_{n,k}\leq 3c\cdot b_{n,k}.\label{refinement-ineq-2}
\end{align}
Therefore, $\cR_{n,k}(c)\subseteq N(C_{1,n}, A'\cdot b_{n,k})\subset \SmallNap$ with $A' = 3c$. 

\begin{figure}
\centering
\includegraphics[scale = 0.55, trim = .5cm 5.5cm 0cm 4cm, clip = TRUE]{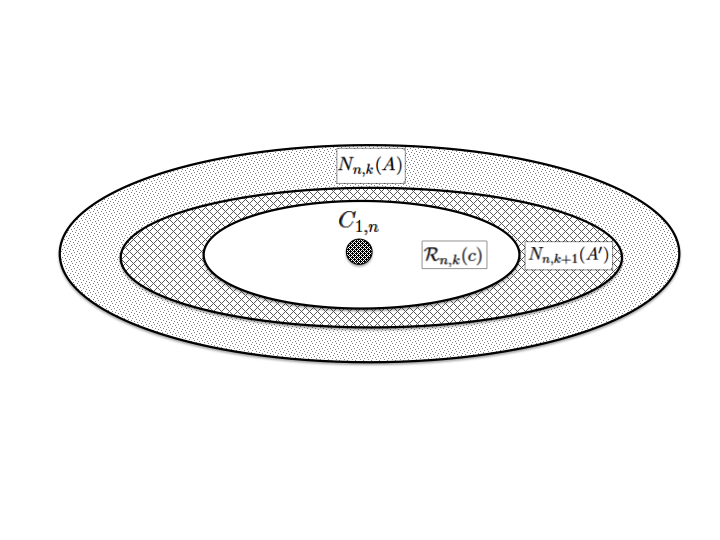}
\caption{\label{fig:proofFig} Illustration of relationship between collections of node sets.}
\end{figure}
We have assumed $\Bopt\in \SmallNa$ with high probability; our aim is to show $\Bopt\in \SmallNap$. Since $\cR_{n,k}(c)\subseteq \SmallNap$, it is sufficient to show that $\Bopt \notin \SmallNa\cap\cR_{n,k}(c)^c$ with high probability. This is illustrated by figure \ref{fig:proofFig}: since $\Bopt$ is inside the outer oval (with high probability), it is sufficient to show that it cannot be outside the inner oval. To this end, it is enough to show that, with high probability, $\widehat H(B, L)<\widehat H(C_{1,n}, L)$ for all sets $B$ in $\SmallNa\cap\cR_{n,k}(c)^c$. Note that by Lemma \ref{conc-cor},
\begin{equation}\label{eq:refinement-lem-9}
\underset{B\in\SmallNa}{\sup}\;\;\widehat H(B, L) < \cH(B, L)+5|L| b_{n,k}
\end{equation}
for large enough $n$, with probability at least $1 - 2\beta_n(\eps)$. Next, since $cb_{n,k}\rightarrow0$ as $n\rightarrow\infty$, by Lemma \ref{prop:optimizer} there exists a constant $a>0$ depending just on $\pi_1$ such that for large enough $n$, $B\in\cR_{n,k}(c)^c$ implies $\cH(B, L)<\cH(C_{1,n}) - a|L|\phi(L)cb_{n,k}$. Applying Lemma \ref{conc-cor} again, we also have $\cH(C_{1,n}, L)<\widehat H(C_{1,n}) + 5|L|b_{n,k}$ with probability at least $1 - 2\beta_n(\eps)$. Furthermore, $\phi(L)\geq\delta^2$ by Lemma \ref{pop-score-rep}. Applying these inequalities to \eqref{eq:refinement-lem-9}, we obtain
\begin{equation}\label{eq:refinement-lem-11}
\underset{B\in\SmallNa \cap \cR_{n,k}(c)^c}{\sup}\;\;\widehat H(B, L) < \widehat H(C_{1,n}, L)- a|L|\delta^2 cb_{n,k} + 10|L| b_{n,k}
\end{equation}
with probability at least $1-4\beta_n(\eps)$. With $c$ large enough, \eqref{eq:refinement-lem-11} implies that $\widehat H(B, L)<\widehat H(C_{1,n}, L)$ for all $B\in \SmallNa \cap \cR_{n,k}(c)^c$. This proves the result in the $\pi_1<\pi_2$ case.

If $\pi_1 = \pi_2$, the argument is almost identical. We instead define $\cR_{n,k}(c)$ as
\[
\cR_{n,k}(c) := \big\{B\subseteq[n]:\max\{|s(B) - \pi_1|,\;\rho(B),\;1 - \rho(B)\}\leq c \cdot b_{n,k}\big\}.
\]
A derivation analogous to that giving inequality \eqref{refinement-ineq-2} yields
\[
n^{-1}\max\{d_h(B, C_{1,n}),\;d_h(B, C_{2,n})\}\leq3c\cdot b_{n,k},
\]
which directly implies that $\cR_{n,k}(c)\subseteq\SmallNap$ with $A' = 3c$. The remainder of the argument follows unaltered.
\end{proof}

\subsubsection{Proof of Theorem \ref{thm:largegraph}}
Recall $Q_\ell(B)$ from Definition \ref{set-modularity} and let $\widehat Q_\ell(B)$ be its random version under the MSBM, as in Section \ref{subsec:score-consistency}. For any $B\subseteq[n]$, we have the inequality
\begin{align}
\left\{\widehat Q_\ell(B)\right\}_+ & \leq \frac{Y_\ell(B)}{n{|B|\choose 2}^{1/2}}\leq \frac{{|B|\choose 2}}{n{|B|\choose 2}^{1/2}} \leq \frac{|B|}{n}.\label{mod-ineq}
\end{align}
This yields the following inequality for $\widehat H(B, L)$:
\begin{align}
\widehat H(B, L) & = |L|^{-1}\left\{\left(\sum_{\ell\in[L]}Q_\ell(B)\right)_+\right\}^2\leq |L|^{-1}\left\{\sum_{\ell\in[L]}Q_\ell(B)_+\right\}^2\leq |L|^{-1}n^{-2}|B|^2.\label{H-ineq}
\end{align}
Recall that $\cB_n(\eps) := \{B\in2^{[n]}: |B|\geq \eps n\}$. Inequality \eqref{H-ineq} implies $\widehat H(B, L)\leq |L|\eps^2$ for all $B\in\cB_n(\eps)^c$. By part 2 of Lemma \ref{pop-score-rep}, $\phi(L)\geq \delta^2$. Therefore, defining $\tau := \frac{\pi_1^2}{2}(1 - \pi_1)^4\delta^2/2$,
\[
|L|\tau<|L|\phi(L)\frac{\pi_1^2}{2}(1 - \pi_1)^4 = H_\ast(B, L).
\]
Therefore, for all $B\in\cB_n(\eps)^c$, we have $\widehat H(B, L) \leq |L|\eps^2<H_\ast(C_{1,n}, L) - |L|(\tau - \eps^2)$. By Lemma \ref{conc-cor}, for large enough $n$ we therefore have
\begin{equation}\label{eq:proof-0}
\sup_{\cB_n(\eps)^c}\widehat H(B,L)< \widehat H(C_{1,n}, L) - |L|(\tau - \eps^2) + 5|L|\gamma_n^{1-\eps}
\end{equation}
with probability greater than $1-2\beta_n(\eps)$, where $\beta_n(\eps):=\exp\{-\frac{\kappa^2\eps}{32}n\gamma_n^{1-\eps}\log n + \log4|L|\}$. For any $\eps<\sqrt{\tau}$, inequality \eqref{eq:proof-0} implies $\widehat H(B, L)<\widehat H(C_{1,n}, L)$ for all $B\in\cB_n(\eps)$, and therefore $\widehat B_{opt}(n)\in\cB_n(\eps)$, with probability at least $1-2\beta_n(\eps)$. Note that $\eps<\sqrt{\tau}<\pi_1$, and $N_{n,k}(\eps) = \cB_n(\eps)$ by Definition \ref{N-ast}. Therefore, the conditions for Lemma \ref{lem:refinement} with $k = 1$ (and $A = \eps$) are satisfied. For any fixed $\eps \in(0,\eta)$ with $\eta:=\sqrt{\tau}$, we may now apply Lemma \ref{lem:refinement} recursively until $1 / 2^k\leq \eps $. This establishes that for sufficiently large $n$,
\begin{equation}\label{eq:proof-2}
\mathbb{P}_n\left(\widehat B_{opt}(n)\in\SmallNa\right) \geq 1 - (2+4k)\beta_n(\eps).
\end{equation}
By definition, $\widehat B_{opt}(n)\in\SmallNa$ implies that
\begin{equation}\label{eq:proof-3}
\mathtt{Error}(\widehat B_{opt}(n)):=\min_{C = C_1, C_2}d_h(\widehat B_{opt}(n), C) \leq A\cdot n \cdot b_{n,k}.
\end{equation}
Since $1/2^k\leq \eps$, note that
\[
n\cdot b_{n,k} = n\gamma_n^{1 - \halfk} = n\cdot n^{\halfk - 1}(\log n)^{1-\halfk} < n^\eps\log n.
\]
Combined with inequality \eqref{eq:proof-2}, this completes the proof.\qed

\subsection{Proof of Theorem \ref{layer-set-consistency}}\label{layer-set-proof}
To prove part (a), we first note that Theorem \ref{thm:largegraph} implies that on the layer set $L_1$, for any $\eps>0$, $\mathtt{Error}(\widehat B_{opt}^{(n)}) = O_p(n^\eps\log n)$. Lemma \ref{pop-score-rep} can be used to show that $\cH(B, L) = 0$ for any $L\subseteq L_0$ and any $B\subseteq[n]$. Using Lemma \ref{prop:sup-score-conc-ineq} and taking a union bound over $L_0$, it is then straightforward to show (using techniques from the proof of Theorem \ref{thm:largegraph}) that on the full layer set $[m]$, for any $\eps>0$, $\mathtt{Error}(\widehat B_{opt}) = O_p(n^\eps\log n)$. Considering now $\widehat L_{opt}^{(n)}$, observe that if $\widehat Q_\ell(B)\leq 0$, then $\widehat H(B,L)=\widehat H(B,L\setminus\{\ell\})$. This immediately implies that $\widehat L_{opt}^{(n)} = \widehat L^+$.

To prove part (b), we note that it is straightforward to show (using Lemma \ref{pop-score-rep}) that $\cH(B,L_1)\geq\cH(B, L)$ for any $L\subset[m]$, with equality if and only if $L = L_1$. Using Lemma \ref{prop:sup-score-conc-ineq} and a union bound over $[m]$ will show that $\widehat L_{opt}^{(n)} = L_1$ with high probability. Applying Theorem \ref{thm:largegraph} completes the part. Part (c) is shown similarly, with the application of Lemma \ref{pop-score-rep} showing that for any $L\subseteq L_1$ and $L'\subseteq[m]$, $\cH(B, L)\geq \cH(B, L')$, with equality if and only if $L'\subseteq L_1$.\qed

\section{The Multilayer Extraction Procedure}\label{sec:algorithm}

The Multilayer Extraction procedure is built around three operations: initialization, extraction, 
and refinement.  In the initialization stage, a family of seed vertex sets is specified.
Next an iterative extraction procedure ({\bf Extraction}) 
is applied to each of the seed sets. {\bf Extraction} alternately updates the layers 
and vertices in a vertex-layer community in a greedy fashion, improving the score at each iteration, 
until no further improvement to the score is possible. 
The family of extracted vertex-layer communities 
is then reduced using the {\bf Refinement} procedure, which ensures that the final collection of 
communities contains the extracted community with largest score, and that the pairwise overlap between any
pair of communities 
is at most $\beta$, where $\beta \in [0,1]$ is a user-defined parameter. {The importance and relevance of this parameter is discussed in Section \ref{sec:beta}}. 
We describe the Multilayer Extraction 
algorithm in more detail below. 

%
%
%

\subsection{Initialization}

For each vertex $u \in [n]$ and layer $\ell \in [m]$ let $N(u,\ell) = \{v \in [n]: \{u,v\} \in E_\ell \}$ be the set of vertices connected to $u$
in $G_\ell$.  We will refer to $N(u,\ell)$ as the neighborhood of $u$ in layer $\ell$. 
Let $\mathcal{B}_0 = \{ N(u,\ell), u \in [n], \ell \in [m] \}$ be the family of all vertex neighborhoods in the observed multilayer network 
$\mathbf{G}(m,n)$.  Multilayer Extraction uses the vertex sets in $\mathcal{B}_0$ as seed sets for identifying communities. Our choice of seed sets is motivated by \citet{gleich2012vertex}, who empirically justified the use of vertex neighborhoods as good seed sets for local detection methods seeking communities with low conductance.


\subsection{Extraction}
Given an initial vertex set, the {\bf Extraction} procedure seeks a vertex-layer community with large score. The algorithm iteratively conducts a {\it Layer Set Search} followed by a {\it Vertex Set Search}, and repeats these steps until a vertex-layer set, whose score is a local maximum, is reached. In each step of the procedure, the score of the candidate community strictly increases, and the procedure is stopped once no improvements to the score are possible. These steps are described next.

\noindent {\it Layer Set Search}: 
For a fixed vertex set $B \subseteq [n]$, {\bf Extraction} searches for the layer set $L$ that maximizes $H(B, \cdot)$ 
using a rank ordering of the layers that depends only on $B$. In particular let $Q_{\ell}(B)$ be the local set modularity of layer $\ell$ from \eqref{set-modularity}. Let $L_o$ be the layer set identified in the previous iteration of the algorithm. {We will now update the layer set $L_0\leadsto L$. This consists of the following three steps}:
\begin{enumeratei}
	\item Order the layers so that $Q_{\ell_1}(B) \geq \cdots \geq Q_{\ell_m}(B)$.
	\item Identify the smallest integer $k$ such that 
$
H(B, \{\ell_1, \ldots, \ell_k\}) \geq H(B, \{\ell_1, \ldots, \ell_k, \ell_{k+1}\}) 
$. {Write $L_p := \set{\ell_1, \ldots, \ell_k}$ for the proposed change in the layer set}.
    \item {\it If} $H(B, L_p) > H(B, L_o)$ set $L = L_p$. {\it Else} set $L = L_o$
\end{enumeratei}


\noindent In the first iteration of the algorithm (where we take $L_o = \emptyset$), 
we set $L = L_p$ in step (iii) of the search. The selected layer set $L_p$ is a local maximum for the score $H(B, \cdot)$. 

%



\noindent {\it Vertex Set Search}: 
Suppose now that we are given a vertex-layer set $(B,L)$. 
{\bf Extraction} updates $B$, one vertex at a time, in a greedy fashion, with updates depending on
the layer set $L$ and the current vertex set. In detail, for each $u \in [n]$ let
\begin{equation}
\label{eq:deltau} 
\delta_u(B,L) = 
\begin{cases} 
H(B / \{u\}, L) - H(B,L) & \mbox{if } \ u \in B\\
H(B \cup \{u\}, L) - H(B,L) & \mbox{if } \ u \notin B .
\end{cases} 
\end{equation}
Vertex Set Search iteratively updates $B$ using the following steps: 
\begin{enumeratei}
	\item Calculate $\delta_u(B,L)$ for all $u \in [n]$.  If $\delta_u(B, L) \leq 0$ for all $u \in [n]$, then stop. 
Otherwise, identify $u^* = \argmax_{u \in [n]} \delta_u(B, L)$.
\item If $u^* \in B$, then remove $u^*$ from $B$.  Otherwise, add $u^*$ to $B$.
\end{enumeratei}

At each iteration of {\bf Extraction}, the score of the updated vertex-layer set strictly increases, and the eventual convergence of this procedure to a local maximum is guaranteed as the possible search space is finite. The resulting local maxima is returned as an extracted community.

\subsection{Refinement}\label{sec:refinement}

Beginning with the $n$ vertex neighborhoods in each layer of the network,
the {\bf Extraction} procedure identifies a collection $\mathcal{C}_T = \{(B_t, L_t)\}_{t \in T}$ 
of at most $m*n$ vertex-layer communities. 
Given an overlap parameter $\beta \in [0,1]$, the family $\mathcal{C}_T$ {is} refined in a greedy fashion, via the {\bf Refinement} procedure,
to produce a subfamily $\mathcal{C}_{S}$, $S \subseteq T$,
of high-scoring vertex-layer sets having the property
that the overlap between any pair of sets is at most $\beta$. 

To quantify overlap, we specify a generalized Jaccard match score to measure overlap between two communities. 
We measure the overlap between two candidate communities $(B_{q}, L_{q})$ and $(B_{r}, L_{r})$ 
using a generalized Jaccard match score
\begin{equation}
\label{eq:match}
		J(q,r) \ = \ \frac{1}{2} \, \dfrac{|B_q \cap B_r|}{|B_q \cup B_r|} \, + \, \frac{1}{2} \, \dfrac{|L_q \cap L_r|}{|L_q \cup L_r|}.
\end{equation} 
It is easy to see that $J(q,r)$ is between 0 and 1.  Moreover, $J(q, r) = 1$ if and only if $(B_q, L_q) = (B_r, L_r)$ and $J(q, r) = 0$ if 
and only if $(B_q, L_q)$ and $(B_r, L_r)$ are disjoint.  Larger values of $J(\cdot, \cdot)$ indicate more overlap between communities.

In the first step of the procedure, {\bf Refinement} identifies and retains the community $(B_s, L_s)$ in 
$\mathcal{C}_{T}$ with the largest score and sets $S = \{s\}$. In the next step, the procedure identifies the community $(B_s, L_s)$ with largest score that satisfies $J(s, s') \leq \beta$ for all $s' \in S$. The index $s$ is then added to $S$. {\bf Refinement} continues expanding $S$ in this way until no further additions to $S$ are possible, namely when for each $s \in T$, there exists an $s' \in S$ such that $J(s, s') > \beta$. The refined collection $C_S = \{B_s, L_s\}_{s \in S}$ is returned.  

\subsubsection{Choice of $\beta$}
\label{sec:beta}
Many existing community detection algorithms have one or more tunable parameters that control the 
number and size of the communities they identify 
\citep{von2007tutorial, leskovec2008statistical, mucha2010community, lancichinetti2011finding, wilson2014testing}.
The family of communities output by Multilayer Extraction depends on the overlap parameter 
$\beta \in [0,1]$. In practice, the value of $\beta$ plays an important role in the structure of the vertex-layer communities. For instance, setting $\beta = 0$ will provide vertex-layer communities that are fully disjoint (no overlap between vertices or layers). On the other hand, 
when $\beta = 1$ the procedure outputs the full set of extracted communities, many of which may be redundant. In exploratory applications, we recommend investigating the identified communities at multiple values of $\beta$, as the structure of communities at different resolutions may provide useful insights about the network itself (see for instance \citet{leskovec2008statistical} or \citet{mucha2010community}). 

Empirically, we observe that the number of communities identified by the Multilayer Extraction procedure is non-decreasing with $\beta$, and there is typically a long interval of $\beta$ values over which the
number and identity of communities remains constant.
In practice we specify a default value of $\beta$ by analyzing the number of communities across a grid of $\beta$ between 0 and 1 in increments of size 0.01. 
For fixed $i$, let $\beta_i = (i-1)*0.01$ and let $k_i =  k(\beta_i)$ denote the number of communities identified at $\beta_i$.  
The default value $\beta'$ is the smallest $\beta$ value in the longest stable window, namely
\[
\beta' \ = \ \mbox{smallest $\beta_i$ such that } k(\beta_i) = \text{mode}(k_1, \ldots, k_{101}). 
\]
%

\section{Application Study}\label{sec:application}

In this section, we assess the performance and potential utility of the Multilayer Extraction procedure through an empirical case study of three multilayer networks, including a multilayer social network, transportation network, and collaboration network. 
We compare and contrast the performance of Multilayer Extraction with six benchmark methods:
Spectral Clustering \citep{newman2006finding}, 
Label Propagation \citep{raghavan2007near},
Fast and Greedy \citep{clauset2004finding},
Walktrap \citep{pons2005computing}, multilayer Infomap \citep{de2014identifying}, and multilayer GenLouvain \citep{mucha2010community, jutla2011generalized}. 
The multilayer Infomap and GenLouvain methods are generalized multilayer methods that can be directly applied to each network considered here. Each of the other four methods have publicly available implementations in the {\it igraph} package in {R} and in Python, and each method is a standard single-layer detection method that can handle weighted edges. We apply the first four methods to both the aggregate (weighted) network computed from the average of the layers in the analyzed multilayer network, and to each layer separately. A more detailed description of the competing methods and their parameter settings is provided in the Appendix. For this analysis and the subsequent analysis in Section \ref{sec:simulations}, we set Multlayer Extraction to identify vertex-layer communities that have a large significance score as specified by equation (\ref{set-modularity}). 



For each method we calculate a number of quantitative features, including the number and size of the identified communities, as well as the number of identified background vertices. For all competing methods, we define a \emph{background vertex} as a vertex that was placed in a trivial community, namely, a community of size one.
We also evaluate the similarity of communities identified by each method. As aggregate and layer-by-layer methods do not provide informative layer information,
we compare the vertex sets identified by each of the competing methods with those identified by 
Multilayer Extraction. To this end, consider two collections of vertex sets $\mathcal{B}, \mathcal{C} \subseteq 2^{[n]}$. Let $\text{size}(\mathcal{B})$ denote the number of vertex sets contained in the collection $\mathcal{B}$ and let $|A|$ represent the number of vertices in the vertex set $A$. Define the coverage of $\mathcal{B}$ by $\mathcal{C}$ as
\begin{equation}
\label{eq:overlap_vertex} 
\text{Co}(\mathcal{B};\mathcal{C}) = \dfrac{1}{\text{size}\left(\mathcal{B}\right)} \sum_{B \in \mathcal{B}} \max_{C \in \mathcal{C}} \left(\dfrac{|B \cap C|}{|B \cup C|} \right).
\end{equation}
The value $\text{Co}(\mathcal{B};\mathcal{C})$ quantifies the extent to which vertex sets in $\mathcal{B}$ are contained {\it in} $\mathcal{C}$.  
In general, $\text{Co}(\mathcal{B};\mathcal{C}) \neq \text{Co}(\mathcal{C};\mathcal{B})$.
The coverage value $\text{Co}(\mathcal{B};\mathcal{C})$ is between 0 and 1, with $\text{Co}(\mathcal{B};\mathcal{C}) = 1$ if and only if 
$\mathcal{B}$ is a subset of $\mathcal{C}$.

We investigate three multilayer networks of various size, sparsity, and relational types: a social network from an Austrailian computer science department \citep{han2014consistent}; an air transportation network of European airlines \citep{cardillo2013emergence}; and a collaboration network of network science authors on arXiv.org \citep{de2014identifying}. The size and edge density of each network is summarized in Table \ref{tab:mult_networks}.


\begin{table}[ht]
	\centering
	\tabcolsep = 0.11cm
	{\small
	\singlespacing
	\begin{tabular}{l c c c}
		Network & $\#$ Layers & $\#$ Vertices & Total $\#$ Edges \\ \hline
		AU-CS & 5 & 61 & 620 \\ 
		EU Air Transport & 36 & 450 & 3588 \\ 
		arXiv & 13 & 14489 & 59026
		\end{tabular}
	}
	\caption[Real Application Summary]{\small Summary of the real multilayer networks in our study. \label{tab:mult_networks}}
	\end{table}
	


\begin{table}[ht]
	\centering
	\tabcolsep = 0.11cm
	\resizebox{\textwidth}{!}{ 
	\begin{tabular}{l| c c c c| c c c c| c c c c}
		\multicolumn{1}{c}{} & \multicolumn{4}{c}{\underline{AU-CS Network}} & \multicolumn{4}{c}{\underline{EU Air Transport Network}} & \multicolumn{4}{c}{\underline{arXiv Network}} \\
		\multicolumn{1}{c}{}  & \multicolumn{4}{c}{} & \multicolumn{4}{c}{} & \multicolumn{4}{c}{} \\
		 & & $\#$ Nodes & & &  & $\#$ Nodes &  & &  & $\#$ Nodes & &\\
		 & Comm. & mean (sd) & Back. & Cov. & Comm. & mean (sd) & Back. & Cov.& Comm. & mean (sd) & Back. & Cov. \\ \hline
		 M-E & 6 & 8.7(2.3) & 11 & 1.00 & 11 & 13.1(5.0) & 358 & 1.00 & 272 &  8.2(3.5) & 12412 & 1.00\\
		 Fast A. & 5 & 12.2(2.6) & 0 & 0.68 & 8 & 52.1(54.6) & 33 & 0.10 & 1543 & 9.1(55.5) & 424 & 0.27 \\
		 Spectral A. & 7 & 8.7(5.4) & 0 & 0.54 &  8 & 52.0(39.6) & 34 & 0.10 & 1435 & 9.8(222.6) & 424 & 0.15\\
		 Walktrap A.& 6 & 10.2(3.3) & 0 & 0.75 & 15 & 22.9(46.5) & 107 & 0.15 & 2238 & 6.3(36.7) & 424 & 0.50\\
		 Label A. & 6 & 10.2(6.5) & 0 & 0.59 & 4 & 104.3(190.5) & 33 & 0.09 & 2329 & 6.0(14.4) & 424 & 0.59\\
		 Fast L. & 6 & 7.7(4.2) & 16.2 & 0.78 & 4 & 13.1(12.4) & 395 & 0.34 & 301 & 6.9(27.8) & 12428 & 0.68\\
		 Spectral L. & 6 & 8.0(5.0) & 16.4 & 0.72 &  3 & 16.1(13.9) & 398 & 0.30 & 283 & 7.2(38.1) & 12457 & 0.66\\
		 Walktrap L.& 7 & 6.6(5.1) & 16.2 & 0.74 & 5 & 10.2(12.8) & 404 & 0.34& 353 & 5.8(15.4) & 12428 & 0.76\\
		 Label L. & 5 & 9.7(9.6) & 16.2 & 0.76 & 2 & 27.8(26.8) & 395 & 0.29 & 383 & 5.4(6.7) & 12428 & 0.79\\
		 GenLouvain & 6 & 50.8(28.4) & 0 & 0.50 & 9 & 49.6(26.1) & 422 & 0.64 & 402 & 9.1(4.5) & 12101 & {0.83}\\
		 Infomap &  20 & 7.0(4.9) & 0 & 0.74 & 41 & 2.3 (2.4) & 33 & 0.21 & 3655 & 7.1(109.1) & 423 & 0.29
		 
	\end{tabular}
	}
	\caption[Summary of Real Application Results]{\small Quantitative summary of the identified communities in each of the three real applications. Shown are the number of communities (Comm.), the mean and standard deviation of the number of nodes in each community, the number of background nodes per layer (Back.) and the coverage of the M-E with the method (Cov.). Methods run on the aggregate network are followed A. and layer-by-layer methods are followed by L. \label{tab:results_mult}}
	\end{table}

Table \ref{tab:results_mult} provides a summary of the quantitative features of the communities identified by each method. For ease of discussion, we will abbreviate Multilayer Extraction by M-E in the next two sections of the manuscript, where we evaluate the performance of the method on real and simulated data.



\subsection{AU-CS Network}
Individuals interact across multiple modes: for example, two people may share a friendship, business partnership, college roommate, or a sexual relationship, to name a few. Thus multilayer network models are particularly useful for understanding social dynamics. In such a model, vertices represent the individuals under study and layers represent the type of interaction among individuals. Correspondingly, vertex-layer communities represent groups of individuals that are closely related across a subset of interaction modes. The number and type of layers in a vertex-layer community describes the strength and type of relationships that a group of individuals share.

We demonstrate the use of M-E on a social network by investigating the AU-CS network, which {describes} online and offline relationships of 61 employees of a Computer Science research department in Australia. The vertices of the 
network represent the employees in the department. The layers of the network represent five different relationships among the employees: 
{\it Facebook}, {\it leisure}, {\it work}, 
{\it co-authorship}, and {\it lunch}.

	\begin{figure}[ht]
		\centering
		\includegraphics[width = 0.8\textwidth, trim = 0cm 0cm 0cm 0cm, clip = TRUE]{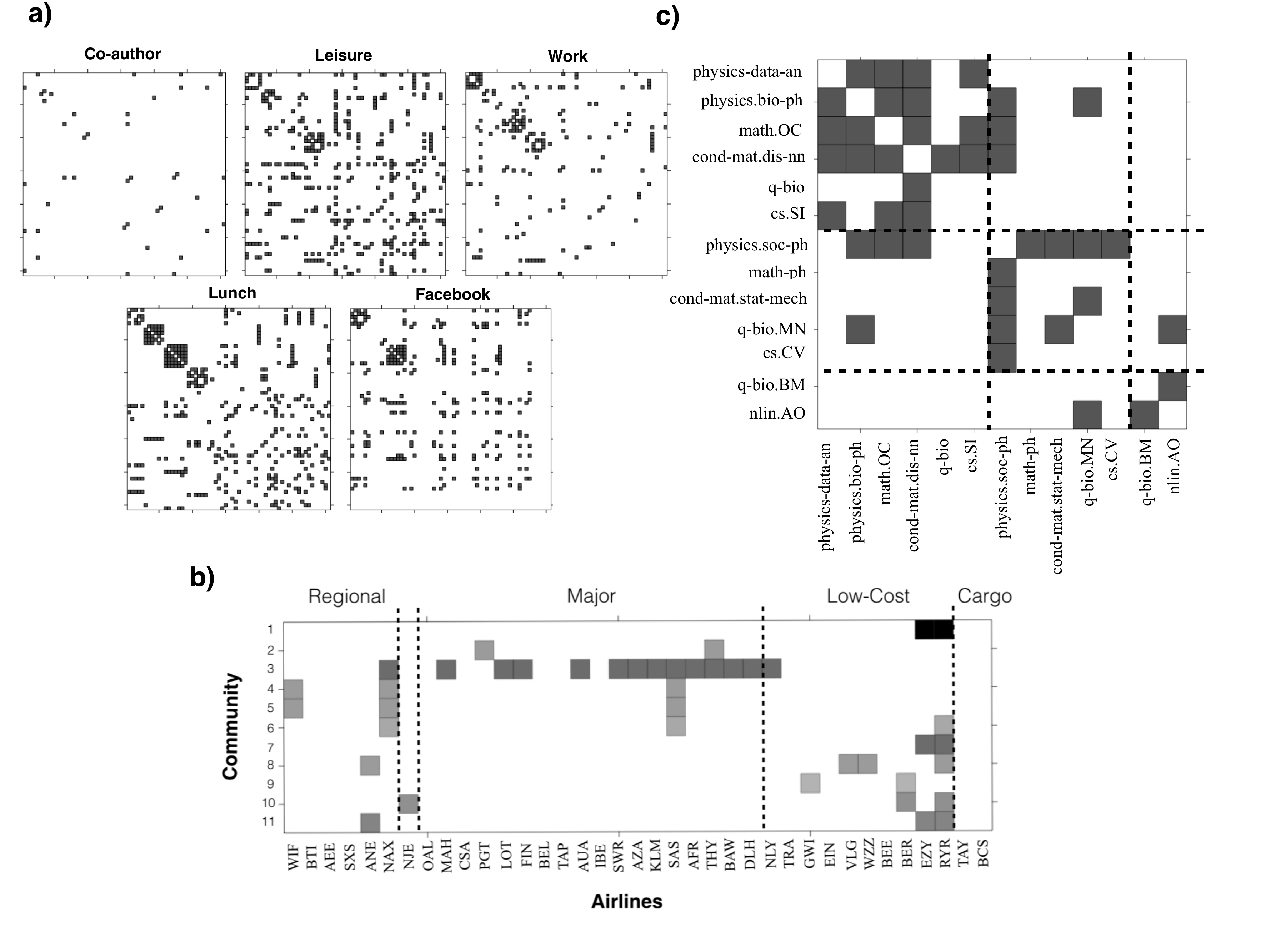}
		\caption[Multilayer Applications]{\small {\bf a)} The AU-CS multilayer network.  The vertices have been reordered based on the six communities identified by M-E.  {\bf b)} The layers of the eleven extracted significant communities identified in the EU transport network. The layers are ordered according to the type of airline. The darkness of the shaded in blocks represents the score of the identified community. {\bf c)} Adjacency matrix of layers in the arXiv network, where edges are placed between layers that were contained in one or more of the communities 
		identified by M-E. Dotted lines separate three communities of submission types that were identified using spectral clustering.}
		\label{fig:mult_applications}
		\end{figure}

\subsubsection*{Results}
M-E identified 6 non-overlapping vertex-layer communities, which are illustrated in Figure \ref{fig:mult_applications} {\bf a}. These communities reveal several interesting patterns among the individuals in the network. Both the {\it work} and {\it lunch} layers were contained in all six of the identified communities, reflecting a natural co-occurrence of work and lunch interactions among the employees. Of the competing methods, GenLouvain was the only other method to identify this co-occurrence.     
Furthermore, two of the identified communities by M-E contained the {\it leisure} and {\it Facebook} layers, both of which are social activities that sometimes extend beyond work and lunch. Thus these communities likely represent groups of employees with a stronger friendship than those that were simply colleagues. These interpretable features identified by M-E provide an example of how our method can be used to provide useful insights beyond aggregate and layer by layer methods. 


With the exception of Infomap, all of the methods identify a similar number of communities (ranging from 5 to 7). Infomap, on the other hand, identified 20 small communities in the multilayer network. The 11 background vertices identified by M-E were sparsely connected, having two or fewer connections in 3 of the layers. As seen by the coverage measure in Table \ref{tab:results_mult}, the vertex sets identified by M-E were similar to those identified by the single-layer methods as well as by Infomap. Furthermore, the communities identified by the aggregate approaches are in fact well contained in the family identified by M-E (average coverage = 0.78). In summary, the vertex sets identified by M-E reflect both the aggregate and separate layer community structure of the network, and the layer sets reveal important features about the social relationships among the employees.

\subsection{European Air Transportation Network}
Multilayer networks have been widely used to analyze transportation networks \citep{strano2015multiplex, cardillo2013emergence}. Here, vertices generally represent spatial locations (e.g., an intersection of two streets, a railroad crossing, GPS coordinates, or an airport) and layers represent transit among different modes of transportion (e.g., a car, a subway, an airline, or bus). The typical aim of multilayer analysis of transportation networks is to better understand the efficiency of transit in the analyzed location. Vertex-layer communities in transportation networks contain collections of vehicles (layers) that frequently travel along the same transit route (vertices). Vertex-layer communities reveal similarity, or even redundancy, in transportation among various modes of transportation and enable optimization of travel efficiency in the location. 

In the present example, we use M-E to analyze the European air transportation network, where vertices represent 450 airports in Europe and layers represent 37 different airlines. An edge in layer $j$ is present between two airports if airline $j$ traveled a direct flight between the two airports on June 1st, 2011. Notably here, each airline belongs to one of five classes: {\it major} (18); {\it low-cost} (10); {\it regional} (6); {\it cargo} (2); and {\it other} (1). A multiplex visualization of this network is shown in Figure \ref{fig:airport}.

	\begin{figure}[ht]
		\centering
		\includegraphics[width = 0.5\textwidth, trim = 0cm 0cm 0cm 0cm, clip = TRUE]{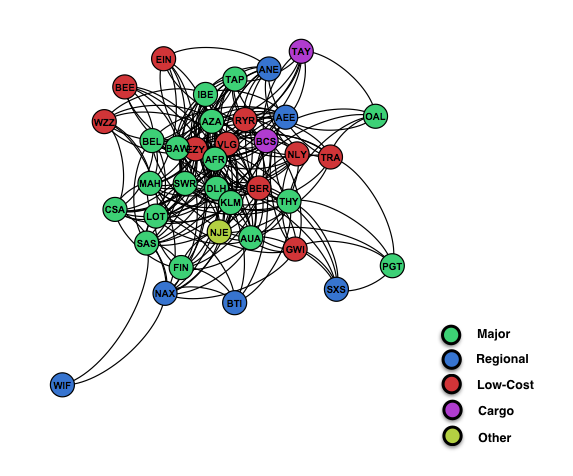}
		\caption[Airport]{\small A one-dimensional visualization of the European Air Transportation Network. Edges are placed between airlines that share at least two routes between airports.}
		\label{fig:airport}
		\end{figure}

\subsubsection*{Results}
The layer sets of the extracted M-E communities are illustrated in Figure \ref{fig:mult_applications} {\bf b}, and a summary of M-E and the competing methods is available in the second major column of Table \ref{tab:results_mult}. M-E identified 11 small communities (mean number of vertices = 13.1, mean number of layers = 3.73). This suggests that the airlines generally follow distinct routes containing a small number of unique airports. Furthermore, Figure \ref{fig:mult_applications} illustrates that the layers of each community are closely associated with airline classes. Indeed, an average of 78 $\%$ of the layers in each community belonged to the same airline class. This reflects the fact that airlines of the same class tend to have direct flights between similar airports. Together these two findings suggest that in general there is little redundancy among the travel of airlines in Europe but that airlines of a similar class tend to travel the same routes. 

Interestingly, the {\it regional} airline Norwegian Air Shuttle (NAX) and the {\it major} airline Scandanavian Airlines (SAS) appeared together in 4 unique communities. These airlines are in fact the top two air carriers in Scandanavia and fly primarily to airports 
in Norway, Sweden, Denmark, and Finland. Thus, M-E reveals that these two airlines share many of the same transportation routes despite the fact that they are different airline classes.

In comparison to the competing methods, we find that both the single-layer methods and M-E identified on the order of 400 background vertices ($\approx 89\%$), which suggests that many of the airports are not frequented by multiple airlines. This finding aligns with the fact that many airports in Europe are small and do not service multiple airline carriers. Aggregate detection approaches, as well as multilayer GenLouvain, identified a similar number of communities as M-E. In fact, the results of M-E most closely matched those of GenLouvain (Coverage = 0.64). We found that Infomap again identified the most communities (almost 4 times as many as any other method) and, like the aggregate approaches, identified few background vertices.

\subsection{arXiv Network}\label{sec:ADHD}

Our final demonstration of M-E is on the multilayer collaboration arXiv network from \citet{de2014identifying}. In a multilayer representation of a collaboration network, vertices represent researchers or other possible collaborators, and layers represent scientific fields or sub-fields under which researchers collaborated. For these applications, multilayer networks provide information about the dissemination and segmentation of scientific research, including the possible overlap of collaborative work under different scientific fields. Vertex-layer communities in collaborative networks represent groups of individuals who collaborated with one another across a subset of scientific fields. Such communities describe how differing scientific fields overlap, which collaborators are working in a similar area, as well as how to disseminate research across fields and how scientists can most easily take part in interdisciplinary collaborations.

The arXiv network that we analyze represents the authors of all arXiv submissions that contained the word ``networks'' in its title or abstract between the years 2010 and 2012. The network has 14489 vertices representing authors, and 13 layers representing the arXiv category under which the submission was placed. An edge is placed between two authors in layer $\ell$ if they co-authored a paper placed in that category. The network is sparse, with each layer having edge density less than 1.5\%.


\subsubsection*{Results}

M-E identified 272 multilayer communities, with an average of 
2.39 layers per community. The communities were small in size, suggesting that network 
science collaboration groups are relatively tightly-knit, both in number of authors and number of differing fields. In Figure \ref{fig:mult_applications} {\bf c}, we plot an adjacency 
matrix for layers whose $(i,j)$ entry is 1 if and only if layers $i$ and $j$ were contained 
in at least one multilayer community. Using the adjacency matrix, the layers of the network
were partitioned into communities using Spectral clustering. This figure identifies the existence of three active interdisciplinary working groups among the selected researchers. These results suggest two important insights. First, network researchers can identify his or her primary sub-field community and best disseminate research in this area by communicating with researchers from other fields in the same community. Second, Figure \ref{fig:mult_applications} {\bf c} illustrates clear separation of three major areas of study. By promoting cross-disciplinary efforts between each of these three major areas, the dissemination of knowledge among network scientists will likely be much more efficient.

Infomap and the aggregate approaches identify on the order of 1000 small to moderately 
sized communities (mean number of vertices between 6.04 and 9.80) with
approximately 423 background vertices. Notably, the 423 background vertices identified by each of these methods had a 0.92 match. On the other hand 
M-E, GenLouvain and the single layer approaches identify a smaller number of communities (between 272 and 402), and classify about 12 thousand 
(roughly $86\%$) of the vertices as background (of which share a match of 0.9).  
These findings suggest that the individual layers of the arXiv network have heterogeneous community structure, and that they contain many non-preferentially attached vertices. Once again, GenLouvain had the highest match with M-E (coverage = 0.83).

\section{Simulation Study}\label{sec:simulations}
 
As noted above, Multilayer Extraction has three key features:
it allows community overlap;
it identifies background; 
and it can identify communities that are present in a small subset of the available layers. 
Below we describe a simulation study that aims to evaluate the performance of M-E with
respect to these features.  The results of additional simulations are described in the Appendix.
 
In this simulation study, we are particularly interested in comparing the performance of M-E with off-the-shelf aggregate and layer-by-layer methods. We make this comparison using the layer-by-layer and aggregate 
approaches using the community detection methods described in Section \ref{sec:application}. 
Define the {\it match} between two vertex families $\mathcal{B}$ and 
$\mathcal{C}$ by
\begin{equation}
\label{eq:total_match}
M( \mathcal{B}; \mathcal{C}) 
\ = \ 
\frac{1}{2} \hskip .2pc \text{Co}( \mathcal{B} ; \mathcal{C}) 
\, + \, 
\frac{1}{2} \hskip .2pc \text{Co}( \mathcal{C} ; \mathcal{B}) ,
\end{equation}
where $\text{Co}( \mathcal{B} ; \mathcal{C} )$ is the coverage measure for vertex families from (\ref{eq:overlap_vertex}). 
The match $M( \mathcal{B}; \mathcal{C})$ is symmetric and takes values in
$[0,1]$.  In particular, 
$M( \mathcal{B}; \mathcal{C}) = 1$ if and only if  
$\mathcal{B} = \mathcal{C}$ and 
$M( \mathcal{B}; \mathcal{C}) = 0$ if and only if
$\mathcal{B}$ and $\mathcal{C}$ are disjoint.

In our simulations, 
we compute the match between the family of vertex sets identified by each method and the family of 
true simulated communities. 
For the layer-by-layer methods, we evaluate the average match of the communities identified in each layer. 
Suppose that $\mathcal{T}$ is the true family of vertex sets in a simulation and $\mathcal{D}$ is 
the family of vertex sets identified by a detection procedure of interest. 
Note that the value $\text{Co}(\mathcal{D}; \mathcal{T})$ quantifies the specificity of $\mathcal{D}$, 
while $\text{Co}(\mathcal{T}; \mathcal{D})$ quantifies its sensitivity;  
thus, $M(\mathcal{D}; \mathcal{T})$ is a quantity between 0 and 1 that summarizes both the sensitivity and specificity of the identified vertex sets $\mathcal{D}$.
The results of the simulation study are summarized in Figures \ref{fig:MSBM}, \ref{fig:persistence}, and \ref{fig:single_embedded} and discussed in more 
detail below.


\subsection{Multilayer Stochastic Block Model}\label{sec:msbm}


In the first part of the simulation study we generated multilayer stochastic block models with $m \in \{1, 5, 10, 15\}$ layers, 
$k \in \{2, 5\}$ blocks, and $n = 1000$ vertices such that each layer has the same community structure.
In more detail, each vertex is first assigned a community label $\{1,\ldots,k\}$ according to a probability mass function 
$\pi = (0.4, 0.6)$ for $k=2$ and $\pi = (0.2,0.1,0.2,0.1,0.4)$ for $k=5$.
In each layer, edges are assigned independently, based on vertex community membership, according to the probability 
matrix $P$ with entries $P(i,i) = r + 0.05$ and $P(i,j) = 0.05$ for $i \neq j$. {Here $r$ is a parameter representing connectivity strength of vertices within the same community.}
The resulting multilayer network consists of $m$ independent realizations of a 
stochastic $k$ block model with the same communities. 
For each value of $m$ and $k$ we vary $r$ from 0.00 to 0.10 in increments of 0.005. 
M-E and all other competing methods are run on ten replications of each simulation. The average match 
of each method to the true communities is given in Figure \ref{fig:MSBM}.

\subsubsection*{Results}

In the single-layer ($m=1$) setting M-E is competitive with 
the existing single-layer and multilayer methods for $r \geq 0.05$, and identifies the true communities without error 
for $r \geq 0.06$. For $m \geq 5$ M-E outperforms all competing single-layer methods 
for $r \geq 0.02$.  As the number of layers increases, M-E and the multilayer Infomap and GenLouvain methods exhibit improved performance 
across all values of $r$. For the two community block model, M-E and the other multilayer methods have comparable performance in networks with five or more layers. M-E outperforms the single-layer and multilayer methods in the $k = 5$ block model when $m > 1$. As expected, aggregate approaches perform well in this simulation, 
outperforming or matching other methods when $m \leq 5$ (results not shown). 
These results suggest that in homogeneous multilayer networks M-E can outperform or match existing methods 
when the network contains a moderate to large number of layers. 

\begin{figure}
		\centering
	\includegraphics[width = \textwidth]{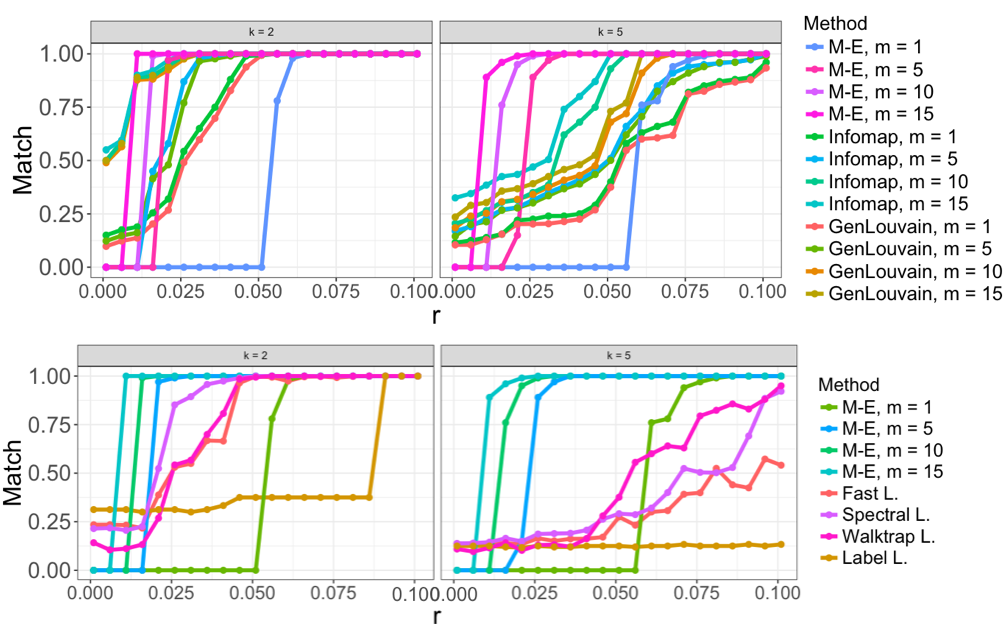}
	\caption[MSBM Simulation Results]{\small (Color) Simulation results for multilayer stochastic block model. In each plot, we report the match of the identified communities  with the true communities where the match is calculated using the match score in (\ref{eq:total_match}).} \label{fig:MSBM}
\end{figure}


\subsection{Persistence} 
In the second part of the simulation study we consider multilayer networks with heterogeneous community structure. 
We simulated networks with 50 layers and 1000 vertices. 
The first $\tau * 50$ layers follow the stochastic block model outlined in Section \ref{sec:msbm} 
with a fixed connection probability matrix $P$ having entries $P(i,i) = 0.15$ and $P(i,j) = 0.05$ for $i \neq j$. 
The remaining $(1 - \tau) * 50$ layers are independent \erdos random graphs with $p = 0.10$, so that in each 
layer every pair of vertices is connected independently with probability $0.10$.  
For each $k \in \{2, 5\}$ we vary the persistence parameter $\tau$ from 0.02 to 1 in increments of 0.02, and  
for each value of $\tau$, we run M-E as well as the competing methods on ten 
replications.  The average match of each method is reported in Figure \ref{fig:persistence}.

\subsubsection*{Results}
\begin{figure}
		\centering
	\includegraphics[width = \textwidth]{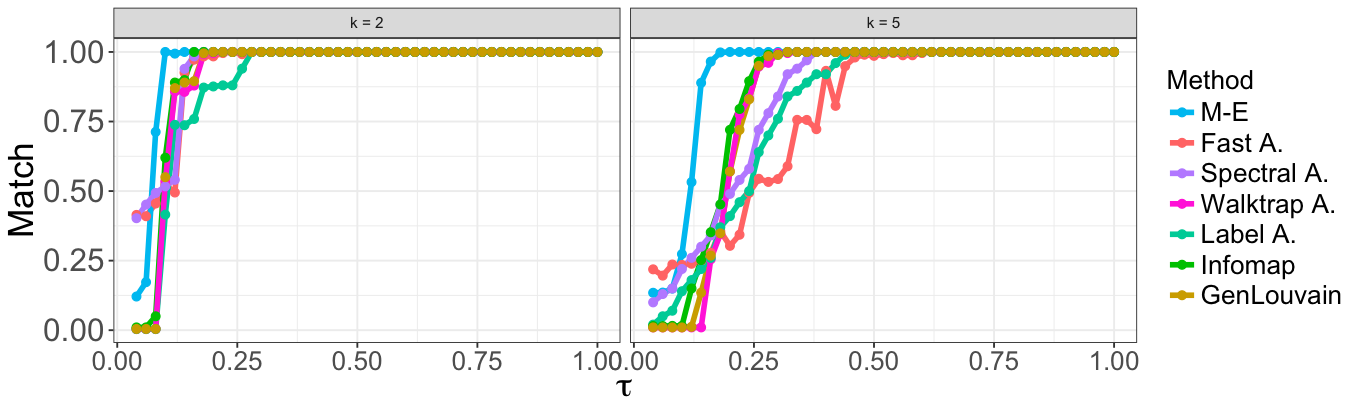}
	\caption[Persistence Simulation Results]{\small (Color) Simulation results persistence simulations. In each plot, we report the match of the identified communities with the true communities where the match is calculated using the match score in (\ref{eq:total_match}).} \label{fig:persistence}
\end{figure}

In both block model settings with $k = 2$ and 5 communities, M-E outperforms competing aggregate and multilayer methods for small values of $\tau$. At these values, aggregate methods perform poorly since the community structure in the layers with signal is hidden by the noisy \erdos layers once the layers are aggregated. Though not shown in Figure \ref{fig:single_embedded}, the layer-by-layer methods are able to correctly identify the community structure of the layers with signal. However, these methods identify on average of 4 or more non-trivial communities in each noisy layer where there is in fact no community structure present. Multilayer GenLouvain and Infomap outperform aggregate methods in this simulation, suggesting that available multilayer methods are able to better handle multilayer networks with heterogeneous layers. Whereas the noisy \erdos layers posed a challenge for both single-layer and aggregate methods, M-E never included any of these layers in an identified community. These results highlight M-E's ability to handle networks with noisy and heterogeneous layers.




\subsection{Single Embedded Communities}

We next evaluate the ability of M-E to detect a single 
embedded community in a multilayer network. 
We construct multilayer networks with $m \in \{1, 5, 10, 15\}$ layers and 1000 vertices according to the following procedure.
Each layer of the network is generated by embedding a common community of size $\gamma * 1000$ in
an \erdos random graph with connection probability $0.05$ in such a way that vertices within the community
are connected independently with probability $0.15$. 
The parameter $\gamma$ is varied between 0.01 and 0.20 in increments of 0.005; ten independent 
replications of the embedded network are generated for each $\gamma$. For each method, we calculate the coverage $C(E, \mathcal{C})$ of the true embedded community $E$ by the identified collection $\mathcal{C}$. We report the average coverage over the ten replications in Figure \ref{fig:single_embedded}.

\subsubsection*{Results}

In the single layer setting, M-E is able to correctly identify the embedded community when the 
embedded vertex set takes up approximately 11 percent ($\gamma = 0.11$) of the layer. 
As before, the performance of M-E greatly improves as the number of layers 
in the observed multilayer network increases. For example at $m =5$ and $m = 10$, the algorithm correctly 
identifies the embedded community (with at least $90\%$ match) once the community has size 
taking as little as 6 percent ($\gamma = 0.055$) of each layer. At $m = 15$, M-E correctly extracts communities with size as small as three percent of the graph in each layer. 

\begin{figure}[ht]
		\centering
	\includegraphics[width = \textwidth]{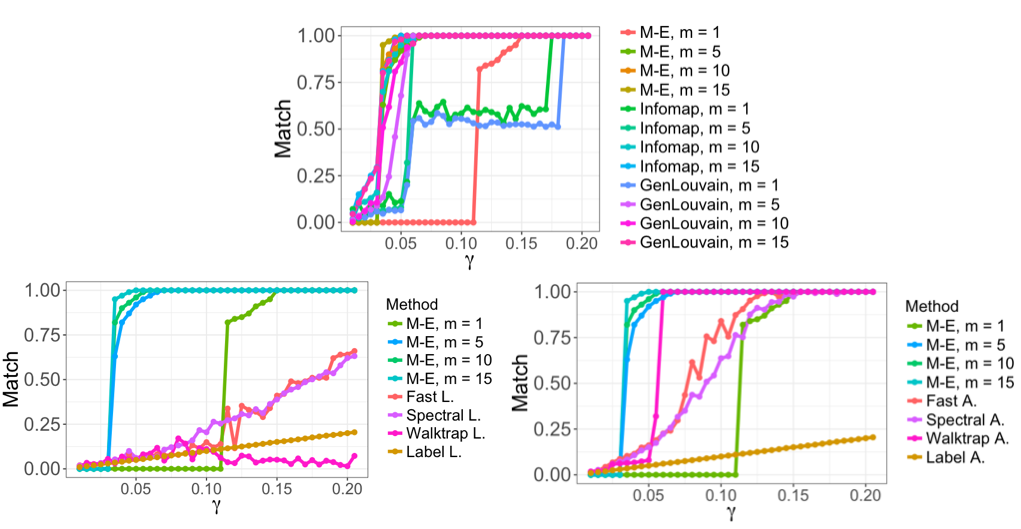} 
	\caption[Single Embedded Community Simulation Results]{\small (Color) Simulation results for single embedded simulations. In each plot, we report the match of the identified communities with the true communities where the match is calculated using the match score in (\ref{eq:total_match}).} \label{fig:single_embedded}
\end{figure}

In the lower right plot of Figure \ref{fig:single_embedded}, we illustrate the results for the aggregate methods applied to the simulated network with $m = 15$ layers. In the lower left plot of Figure \ref{fig:single_embedded}, we show the results for the layer-by-layer methods in this simulation and the upper part of the figure show the results for the multilayer methods. For $m \geq 5$ M-E outperforms all of the competing aggregate methods. In addition, M-E outperforms every layer-by-layer method for all $m$. Multilayer GenLouvain and Infomap have comparable performance to M-E across all $m$ and $\gamma$. These results emphasize the extraction cabilities of M-E and show that the procedure, contrary to aggregate and single-layer competing methods, is able to detect small embedded communities in the presence of background vertices.

\section{Discussion}\label{sec:discussion}

Multilayer networks have been profitably applied to a number of complex systems, and community detection is a valuable exploratory technique to analyze and understand networks. In many applications, the community structure of a multilayer network will differ from layer to layer due to heterogeneity. In such networks, actors interact in tightly connected groups that persist across only a subset of layers in the network. In this paper we have introduced and evaluated the first community detection method to address multilayer networks with heterogeneous communities, Multilayer Extraction. The core of Multilayer Extraction is a significance score that quantifies the connection strength of a vertex-layer set by comparing connectivity in the observed network to that of a fixed degree random network. 

Empirically, we showed that Multilayer Extraction is able to successfully identify communities in multilayer networks with overlapping, disjoint, and heterogeneous community structure. Our numerical applications revealed that Multilayer Extraction can identify relevant insights about complex relational systems beyond the capabilities of existing detection methods. We also established asymptotic consistency of the global maximizer of the Multilayer Extraction score under the multilayer stochastic block model. We note that in practice the Multilayer Extraction procedure can identify overlapping community structure in multilayer networks; however, our consistency results apply to a multilayer model with non-overlapping communities. Future work should investigate consistency results like the ones derived here on a multilayer model with overlapping communities, such as a multilayer generalization of the mixed membership block model introduced in \citet{airoldi2008mixed}. We expect that similar theory will hold in an overlapping model and that the theoretical techniques utilized here can be used to prove such results.

The Multilayer Extraction method provides a first step in understanding and analyzing multilayer networks with heterogeneous community structure. This work encourages several interesting areas of future research. For instance, the techniques used in this paper could be applied, with suitable models, to networks having ordered layers (e.g. temporal networks), as well as to networks with weighted edges such as the recent work done in \citet{palowitch2016continuous, wilson2017stochastic}. Furthermore, one could incorporate both node- and layer-based covariates in the null model to handle exogenous features of the multilayer network. Finally, it would be interesting to evaluate the consistency of Multilayer Extraction in multilayer networks in the high dimensional setting where the number of vertex-layer communities grows with the number of vertices. 

\section*{Acknowledgements}
The authors gratefully acknowledge Peter Mucha for helpful discussions and suggestions for this work. We also thank the associate editor and the two anonymous referees whose comments greatly improved the content and exposition of the paper. The work of JDW was supported in part by NSF grants DMS-1105581, DMS-1310002, and SES grant 1357622. The work of JP was supported in part by NIH/NIMH grant R01-MH101819-01. The work of SB was supported in part by NSF grants DMS-1105581, DMS-1310002, DMS-160683, DMS-161307, SES grant 1357622, and ARO grant W911NF-17-1-0010. The work of ABN was supported in part by NSF DMS-1310002, NSF DMS-1613072, as well as NIH HG009125-01, NIH MH101819-01. The content is solely the responsibility of the authors and does not necessarily represent the official views of the National Institutes of Health.

\appendix
\section{Proofs of Lemmas from Section \ref{subsec:score-consistency}} \label{sec:AppA}




\subsection{Proof of Lemma \ref{pop-score-rep}}
It is easy to show that for any $2\times 2$ symmetric matrix $A$ and 2-vectors $\mathbf x$, $\mathbf y$,
\[
(\mathbf x^TA\mathbf x)(\mathbf y^TA\mathbf y) - (\mathbf x^TA\mathbf y)^2 = (\mathbf x_1\mathbf y_2 - \mathbf x_2\mathbf y_1)^2\text{det}(A).
\]
Fix $B\subseteq[n]$ and let $s,\rho$, and $v$ correspond to $B$, as in Definition \ref{rho}. Then for any $\ell\in[L]$, using the fact that $\kappa_\ell := \pi^TP_\ell\pi$ and the identity above, we have
\begin{align*}
\Bv^tP_\ell \Bv - \dfrac{(\pi^t P_\ell \pi)^2}{\kappa_\ell} & = \frac{\kappa_\ell \Bv^tP_\ell \Bv}{\kappa_\ell} - \dfrac{(\Bv^t P_\ell \pi)^2}{\kappa_\ell} = \dfrac{(\pi^tP_\ell\pi)(\Bv^tP_\ell \Bv) - (\Bv^t P_\ell \pi)^2}{\kappa_\ell}\\[1em]
& = (\pi_1(1-\rho) - \pi_2\rho)^2\dfrac{\det{P_\ell}}{\kappa_\ell} = (\pi_1 - \rho)^2\dfrac{\det{P_\ell}}{\kappa_\ell}.
\end{align*}
Recall that $q_\ell(B) :=\frac{s}{\sqrt{2}}\left(\Bv^tP_\ell \Bv - (\pi^t P_\ell \pi)^2/\kappa_\ell\right)$ and $H_\ast(B, L) = |L|^{-1}\left(\sum_\ell q_\ell(B)\right)^2$. Part 1 follows by summation over $L$. For part 2, note that\\
\noindent $\pi_1^2P_\ell(1,1) + \pi_2^2P(2,2)\geq 2\pi_1\pi_2\sqrt{P_\ell(1,1)P_\ell(2,2)}$. Therefore,
\begin{align*}
\kappa_\ell & = \pi_1^2P_\ell(1,1) + 2\pi_1\pi_2P_\ell(1,2) + \pi_2^2P(2,2)\\[1em]
& \geq 2\pi_1\pi_2\left(\sqrt{P_\ell(1,1)P_\ell(2,2)} + P_\ell(1,2)\right)\\[1em]
& \geq 2\pi_1\pi_2\left(\sqrt{P_\ell(1,1)P_\ell(2,2)} + P_\ell(1,2)\right)\left(\sqrt{P_\ell(1,1)P_\ell(2,2)} - P_\ell(1,2)\right)\\[1em]
& = 2\pi_1\pi_2\delta \geq \pi_1\delta.
\end{align*}
Thus $\delta\leq\frac{\det P_\ell}{\kappa_\ell}\leq \frac{1}{\pi_1\delta}$. Part 2 follows.\qed




\subsection{Proof of Lemma \ref{prop:optimizer}}
Define $g:2^{[n]}\mapsto\mathbb{R}$ by $g(B) := \frac{s(B)}{2}(\rho(B) - \pi_1)^4$. Recall the function $\phi(L)$ defined in Lemma \ref{pop-score-rep}. Note that part 1 of Lemma \ref{pop-score-rep} implies $H_\ast(B, L) = |L|\phi(L)g(B)$. It is therefore sufficient to show that there exists a constant $a>0$ such that for sufficiently small $t$, $B\in\cR(t)^c$ implies $g(B) < g(C_{1,n}) - at$. We will show this separately for the $\pi_1<\pi_2$ and $\pi_1=\pi_2$ cases.\\

\underline{\emph{Part 1} ($\pi_1<\pi_2$)}: Define the intervals $I_1 := [0,\pi_1]$, $I_2:=(\pi_1,\pi_2]$, and $I_3:=(\pi_2,1]$. We trisect $2^{[n]}$, the domain of $g$, with the collections $\cD_{i,n} := \{B\subseteq[n]:s(B)\in I_i\}$, for $i = 1,2,3$. We will prove that the inequality $g(B) < g(C_{1,n}) - at$ holds for all $B\in\cR(t)$ on each of those collections. We will continually rely on the fact that $B\in\cR(t)$ implies at least one of the inequalities (I) $|s(B)-\pi_1|>t$ or (II) $1-\rho(B)<t$ is true.

Suppose $B\in\cR(t)^c\cap\cD_{1,n}$ and inequality (I) is true. Then $s(B)<\pi_1-t$, and
\begin{align}
g(B) := \frac{s(B)^2}{2}(\rho(B) - \pi_1)^4 & \leq \frac{s(B)^2}{2}(1 - \pi_1)^4 \text{ (since $\pi_1\leq 1/2$)}\nonumber\\
& < \frac{(\pi_1 - t)^2}{2}(1 - \pi_1)^4= \frac{\pi_1^2}{2}(1 - \pi_1)^4 - 2t(1 - \pi_1)^4 + o(t)\nonumber\\
& < \frac{\pi_1^2}{2}(1 - \pi_1)^4 - t(1 - \pi_1)^4 = g(C_{1,n}) - t(1 - \pi_1)^4\label{zero-ineq}
\end{align}
for sufficiently small $t$. If inequality (II) is true, then
\begin{align*}
(\rho(B) - \pi_1)^4 \leq \max\{(1-t-\pi_1)^4,\pi_1^4\} & = \max\{(\pi_2 - t)^4,\pi_1^4\} = (\pi_2 - t)^4
\end{align*}
for sufficiently small $t$, as $\pi_1<\pi_2$. Therefore,
\begin{equation}\label{zero-ineq2}
g(B) \leq \frac{\pi_1^4}{2}(\pi_2 - t)^4 = \frac{\pi_1^4}{2}\pi_2^4 - 4\pi_2^3t + o(t) < g(C_{1,n}) - 2\pi_2^3t
\end{equation}
for sufficiently small $t$. Thus for all $B\in\cR(t)^c\cap\cD_{1,n}$, $g(B)<g(C_{1,n}) - a_1t$ with $a_1 = \min\{(1 - \pi_1)^4,2\pi_2^3\}$.

Suppose $B\in\cR(t)^c\cap\cD_{2,n}$ and inequality (I) is true. Then $s(B)>\pi_1 + t$. Note that $0\leq\rho(B)|B|\leq |C_{1,n}|$, yielding the useful inequality
\begin{equation}\label{useful-ineq1}
0\leq \rho(B)\leq \pi_1/s(B).
\end{equation}
Subtracting through by $\pi_1$ gives 
\[
(\rho(B) - \pi_1)^4 \leq\max\{\pi_1^4,\pi_1^4(1/s(B) - 1)^4\} = \pi_1^4(1/s(B) - 1)^4.
\]
Therefore,
\begin{align}\label{D2-1-ineq1}
g(B)\leq\frac{s(B)^2}{2}\pi_1^4(1/s(B) - 1)^4 &= \frac{\pi_1^4}{2}(1/\sqrt{s(B)} - \sqrt{s(B)})^4 \nonumber \\[1em]
& < \frac{\pi_1^4}{2}(1/\sqrt{\pi_1 + t} - \sqrt{\pi_1 + t})^4,
\end{align}
since $F(x) := (1/\sqrt{x} - \sqrt{x})^4$ is decreasing on $(0,1]$, and $s(B)>\pi_1 + t$. Note that 
\begin{equation}\label{taylor-derivative1}
\frac{d}{dt}\left(\frac{1}{\sqrt{\pi_1 + t}} - \sqrt{\pi_1 + t}\right)^4 = -3\left(\frac{1}{\sqrt{\pi_1 + t}} - \sqrt{\pi_1 + t}\right)^3\left(\frac{1}{2(\pi_1+t)^{3/2}} + \frac{1}{2\sqrt{\pi_1 + t}}\right).
\end{equation}
By Taylor's theorem, this implies that
\[
(1/\sqrt{\pi_1 + t} - \sqrt{\pi_1 + t})^4 = (1/\sqrt{\pi_1} - \sqrt{\pi_1})^4 - a_2t + o(t) < (1/\sqrt{\pi_1} - \sqrt{\pi_1})^4 - a_2t/2
\]
for sufficiently small $t$, where $a_2$ is the right-hand-side of \eqref{taylor-derivative1} at $t = 0$. Note further that $(1/\sqrt{\pi_1} - \sqrt{\pi_1})^4 = (\pi_2/\sqrt{\pi_1})^4 = \pi_2^4/\pi_1^2$. Putting these facts together with inequality \eqref{D2-1-ineq1}, we obtain
\begin{equation}\label{D2-1-ineq-final}
g(B) < \frac{\pi_1^4}{2}\frac{\pi_2^4}{\pi_1^2} - a_2t/2 = \frac{\pi_1}{2}\pi_2^4 - a_2t/2 = g(C_{1,n}) - a_2t/2.
\end{equation}
If inequality (II) is true, $\rho(B) < 1 - t$. If $\rho(B)\leq\pi_1$, $(\rho(B) - \pi_1)^4$ is maximized when $\rho(B) = 0$, so that, since $s(B)\leq\pi_2$,
\begin{equation}\label{D2-2-ineq-final}
g(B) \leq \frac{\pi_2^2}{2}\pi_1^4 = g(C_{1,n}) + \frac{\pi_2^2}{2}\pi_1^4 - \frac{\pi_1^2}{2}\pi_2^4 = g(C_{1,n}) + \frac{\pi_1^2\pi_2^2}{2}(\pi_2^2 - \pi_1^2) < g(C_{1,n}) - t
\end{equation}
for sufficiently small $t$, since $\pi_1$ is fixed. If $\rho(B) >\pi_1$, note that inequality \eqref{useful-ineq1} implies $s(B)\leq\pi_1/s(B)$. Therefore,
\begin{align}\label{D2-3-ineq1}
g(B)\leq\frac{\pi_1^2}{2\rho(B)^2}(\rho(B) - \pi_1)^4 &= \frac{\pi_1^2}{2}\left(\sqrt{\rho(B)} - \pi_1/\sqrt{\rho(B)}\right)^4\nonumber \\[1em]
& < \frac{\pi_1^2}{2}\left(\sqrt{1-t} - \pi_1/\sqrt{1-t}\right)^4
\end{align}
since $G(x):=(\sqrt{x} - \pi_1/\sqrt{x})^4$ is increasing on $(\pi_1,1]$. A similar Taylor expansion argument to that yielding inequality \eqref{D2-1-ineq-final} yields, for a constant $a_3$ depending only on $\pi_1$,
\begin{equation}\label{D2-3-ineq-final}
g(B) < \frac{\pi_1^2}{2}(1 - \pi_1)^4 - a_3t/2 = g(C_{1,n}) - a_3t/2,
\end{equation}
for sufficiently small $t$. Pulling together inequalities \eqref{D2-1-ineq-final}, \eqref{D2-2-ineq-final}, and \eqref{D2-3-ineq-final}, we have that for all $B\in\cR(t)^c\cap\cD_{1,n}$, $g(B) < g(C_{1,n}) -a_4$ with $a_4:=\min\{a_2/2,1,a_3/2\}$.

Suppose $B\in\cR(t)^c\cap\cD_{3,n}$. Note that $|B|-|C_{2,n}|\leq |B\cap C_{1,n}|\leq|C_{1,n}|$. Dividing through by $|B|$ yields the useful inequality
\begin{equation}\label{useful-ineq2}
1 - \pi_2/s(B)\leq\rho(B)\leq \pi_1/s(B).
\end{equation}
Subtracting inequality \eqref{useful-ineq2} by $\pi_1$ gives 
\[
\pi_2(1 - 1/s(B))\leq \rho(B) - \pi_1^4 \leq \pi_1(1/s(B) - 1).
\]
Since $\pi_1<\pi_2$, this implies that $(\rho(B) - \pi_1)^4 \leq \pi_2^4(1 - 1/s(B))^4$. Therefore,
\begin{equation}\label{D3-1-ineq1}
g(B)\leq\frac{s(B)^2}{2}\pi_2^4(1/s(B) - 1)^4 = \frac{\pi_2^4}{2}\left(1/\sqrt{s(B)} - \sqrt{s(B)}\right)^4<\frac{\pi_2^4}{2}(1/\sqrt{\pi_2} - \sqrt{\pi_2})^4,
\end{equation}
since $F(x):=(1\sqrt{x} - \sqrt{x})^4$ is decreasing on $I_3:=(\pi_2, 1]$ and $s(B)\in I_3$. Note that $\sqrt{\pi_2} - 1 / \sqrt{\pi_2}=-\pi_1/\sqrt{\pi_2}$. Therefore,
\[
g(B) < \frac{\pi_2^4}{2}\frac{\pi_1^4}{\pi_2^2} = \frac{\pi_2^2}{2}(0 - \pi_1)^4 = g(C_{2,n}) < g(C_{1,n}) - t
\]
for $t$ sufficiently small. Thus, for $a:=\min\{a_1, a_4, 1\}$, for sufficiently small $t$ we have $g(B) < g(C_{1,n}) - at$ whenever $B\in\cR(t)$. This completes the proof in the case $\pi_1 < \pi_2$.\\

\underline{\emph{Part 2} ($\pi_1=\pi_2$)}: Recall that when $\pi_1 = \pi_2$ we define $\cR(t)$ by
\[
\cR(t) := \big\{B\subseteq[n]: \max\{|s(B) - \pi_1|,\;\rho(B),\; 1  - \rho(B)\} \leq t\big\}.
\]
Hence we will use the fact that $B\in\cR(t)$ implies at least one of the inequalities (I) $|s(B)-\pi_1|>t$ or (II) $t<\rho(B)<1-t$ is true. Define the intervals $I_1 := [0,\pi_1]$, $I_2:=(\pi_1,1]$. We bisect $2^{[n]}$, the domain of $g$, with the collections $\cD_{i,n} := \{B\subseteq[n]:s(B)\in I_i\}$, for $i = 1,2$. We will prove that the inequality $g(B) < g(C_{1,n}) - at$ holds for all $B\in\cR(t)$ on each of those collections. 

Suppose $B\in\cR(t)^c\cap\cD_{1,n}$ and inequality (I) is true. Then the same derivation yielding inequality \eqref{zero-ineq} gives $g(B)<g(C_{1,n}) - t(1 - \pi_1)^4$ for sufficiently small $t$. If inequality (II) is true, then
\begin{align*}
(\rho(B) - \pi_1)^4 \leq \max\{(1-t-\pi_1)^4,(\pi_1-t)^4\} = \max\{(\pi_2-t)^4,(\pi_1-t)^4\} = (\pi_2 - t)^4,
\end{align*}
since $\pi_1 = \pi_2$. Therefore, inequality \eqref{zero-ineq2} remains intact. Both inequalities hold on $I_2$ as well, for the roles of $\pi_1$ and $\pi_2$ may be interchanged, and the derivations treated symmetrically. This completes the proof in the case $\pi_1 = \pi_2$.\qed




\subsection{Proof of Lemma \ref{prop:sup-score-conc-ineq}}
Recall the definitions of set modularity and population set modularity from Definitions \ref{set-modularity} and \ref{pop-mod}. Define $W := \sum_{\ell\in[L]}\widehat Q_\ell(B)$ and $w := \sum_{\ell\in[L]}\cQ_\ell(B)$. Note that by Part 1 of Lemma \ref{pop-score-rep}, $q_\ell(B)\geq0$ regardless of $B$, a fact which will allow the application of Lemma \ref{pos-sum-fact} in what follows. We have $\widehat H(B, L) = |L|^{-1}W_+^2$, $\cH(B, L) = |L|^{-1}w^2$, and for any $B$ such that $|B|\geq n\eps$,
\begin{align*}
& \mathbb{P}_n\left(\big|\widehat H(B, L) - \cH(B, L)\big|>\dfrac{4|L|t}{n^2} + \frac{52|L|}{\kappa n}\right) = \;\mathbb{P}_n\left(\big|W_+^2 - w^2|>\dfrac{4|L|^2t}{n^2} + \frac{52|L|^2}{\kappa n}\right)\\[1em]
& \leq \;\mathbb{P}_n\left(\max_{\ell\in[L]}\big|\widehat Q_\ell(B) - \cQ_\ell(B)\big|>\dfrac{t}{n^2} + \frac{13}{\kappa n}\right) \leq \;4|L|\exp\left(-\kappa^2 \dfrac{\eps t^2}{16n^2}\right)
\end{align*}
for large enough $t > 0$, where the first inequality follows from Lemma \ref{pos-sum-fact} for large enough $n$, and the second inequality follows from Lemma \ref{lem:sup-score-conc-ineq0} and a union bound. Applying a union bound over sets $B\in\cB_n$ yields the result.\qed




\subsection{Proof of Lemma \ref{conc-cor}}
Assume first that $k > 1$. By definition, $B\in \SmallNa$ implies that at least one of $d_h(B, C_1)\leq A\cdot n\cdot b_{n,k-1}$ or $d_h(B, C_2)\leq A\cdot n\cdot b_{n,k-1}$ is true. Suppose the first inequality holds. Since $d_h(B, C_1) = |B\setminus C_1| + |C_1\setminus B|$, we have the inequality
\begin{align*}
\big||B| - n\pi_1\big| = \big||B| - |C_{1}|\big| & \leq \big||B| - |B\cap C_{1}| - |C_{1}| + |B\cap C_{1}|\big| \leq \big||B| - |B\cap C_{1}|\big|\\[1em]
& + \big||C_{1}| - |B\cap C_{1}|\big| = |B\setminus C_{1}| + |C_{1} \setminus B|\leq A\cdot n\cdot b_{n,k-1}.
\end{align*}
Alternatively, if $d_h(B, C_2)\leq A\cdot n\cdot b_{n,k-1}$, we have the same bound for $\big||B| - n\pi_2\big|$. Therefore, since $\pi_1\leq\pi_2$, $B\in \SmallNa$ implies that $|B|\geq n\pi_1-A\cdot n\cdot b_{n,k-1}$. Since $b_{n,k-1} = o(1)$ as $n\rightarrow\infty$ and $\eps<\pi_1$, this implies that for large enough $n$, $\SmallNa\subseteq\cB_n(\eps)$. By Lemma \ref{prop:sup-score-conc-ineq}, therefore, for large enough $n$, we have
\begin{align*}
& \mathbb{P}_n\left(\sup_{\SmallNa}\Big|\widehat{H}(B, L) - \cH(B, L)\Big| > \dfrac{4|L|t}{n^2} + \frac{52|L|}{\kappa n}\right)\leq 4|L||\SmallNa| \exp\left(-\kappa^2\dfrac{\eps t^2}{16n^2}\right)\addtocounter{equation}{1}\tag{\theequation}\label{eq:refinement-lem-1}
\end{align*} 
for all $t>0$. We now bound the right-hand side of inequality \eqref{eq:refinement-lem-1} with $t$ replaced by $t_n: = n^{1 + \halfk}(\log n)^{1 -\halfk}$. Note that 
\[
\frac{t_n^2}{n^2} \ =\ \frac{1}{n^2}n^{2+\halfkMO}(\log n)^{2-\halfkMO} \ =\ n\cdot n^{\halfkMO-1}(\log n)^{1-\halfkMO}\log n \ =\ n\cdot b_{n,k-1}\log n.
\]
Furthermore, by Corollary \ref{special_set_size} (see Appendix \ref{technical-results}) we have $|\SmallNa|\leq 2\exp\left[3A\cdot n\cdot b_{n,k-1}\log\left(1/b_{n,k-1}\right)\right]$. These facts yield the bound
\begin{align*}
|\SmallNa| \exp\left(-\kappa^2\dfrac{\eps t_n^2}{16n^2}\right) & \leq 2\exp\left\{-\kappa^2\frac{\eps}{16}n\cdot b_{n,k-1}\left(\log n - \frac{16}{\kappa^2\eps}3A\log\left(1/b_{n,k-1}\right)\right)\right\}\\[1em]
& \leq 2\exp\left(-\kappa^2\frac{\eps}{32}n\cdot b_{n,k-1}\log n\right)\text{ (for large $n$, since $1/b_{n,k-1} = o(n)$)}\\[1em]
& < 2\exp\left( -\kappa^2\frac{\eps}{32}n\gamma_n^{1-\eps}\log n\right),
\end{align*}
where the final inequality follows from the choice of $k$ satisfying $\halfkMO<\eps$. Therefore,
\begin{equation}\label{eq:refinement-lem-2}
4|L||\SmallNa|\exp\left(-\kappa^2\dfrac{\eps t_n^2}{16n^2}\right) \leq 2\exp\left\{-\frac{\kappa^2\eps}{32}n\gamma_n^{1-\eps}\log n + O(\log|L|)\right\}
\end{equation}
for large enough $n$. Notice now that $t_n/n^2 = b_{n,k}$ vanishes slower than $1/n$, and is therefore the leading order term in the expression $\frac{4|L|t_n}{n^2}+\frac{52|L|}{\kappa n}$ (see equation \ref{eq:refinement-lem-1}). Hence for large enough $n$ we have $\frac{4|L|t_n}{n^2}+\frac{52|L|}{\kappa n}\leq 5|L|b_{n,k}$. Combining this observation with lines \eqref{eq:refinement-lem-1} and \eqref{eq:refinement-lem-2} proves the result in the case $k > 1$.

If $k = 1$, assume $A = \eps$. By definition, then (see Definition \ref{N-ast}), $N_{n,k}(A) = \cB_n(\eps)$. Returning to inequality \eqref{eq:refinement-lem-1}, we note that $\log|\cB_n(\eps)| = O(n)$, and thus we can derive the bound \eqref{eq:refinement-lem-2} with the same choice of $t_n := n^{1 + \halfk}(\log n)^{1 -\halfk} = n\sqrt{n\log n}$. The rest of the argument goes through unaltered.\qed





\section{Technical Results}\label{technical-results}



		\begin{lemma}
			\label{fact:close-sets-complexity}
		      Fix $\pi_1\in[0, 1]$. For each $n$, let $C_1\subseteq[n]$ be an index set of size $\lfloor n\pi_1 \rfloor$. Let $C_2 := [n]\setminus C_1$. Let $\gamma_n\in[0,1]$ be a sequence such that $\gamma_n \rightarrow 0$ and $n\gamma_n\rightarrow\infty$. Then for large enough $n$,
		      \[
		      |N(C_1,\gamma_n)| \leq \exp\{3n\gamma_n\log(1/\gamma_n)\}.
		      \]
\end{lemma}

\begin{proof}
Define the boundary of a neighborhood of $C\subseteq[n]$ by
\[
\partial N(C,r) := \{B\subseteq[n]:d_h(B,C) = \lfloor nr\rfloor\}.
\]
Note that any $B\subseteq[n]$ may be written as the disjoint union $B = \{C_2\cap B\}\cup\{C_1\cap B\}$. Since $C_1\cap B = C_1\setminus\{C_1\setminus B\}$, for fixed $k\in[n]$ it follows that each set $B\in \partial N(C,k/n)$ is uniquely identified with choices of $|C_2\cap B|$ indices from $C_2$ and $|C_1\setminus B|$ indices from $C_1$ such that
\[|B\cap C_2| + |C_1\setminus B| = |B \setminus C_1| + |C_1\setminus B| = d_h(B, C_1) = k.
\]
Therefore, we have the equality
\begin{equation}\label{perimeter-eq}
|\partial N(C_1,k)|
\;=\;
\sum_{m = 0}^{k}\left\{{|C_2|\choose m} + {|C_1|\choose k - m}\right\}.
\end{equation}
Note that for positive integers $K, N$ with $K < N/2$, properties of the geometric series yield the following bound:
\begin{align}
{N\choose K}^{-1}\sum_{m=0}^K{N\choose m} & = \sum_{m = 0}^K\frac{(N-K)!K!}{(N-m)!m!} = \sum_{m = 0}^K\prod_{j = m + 1}^{K}\frac{j}{N - j + 1}\nonumber\\[1em]
& < \sum_{m = 0}^K\left(\frac{K}{N - K + 1}\right)^m < \frac{N - (K - 1)}{N - (2K - 1)}.\label{geoseries}
\end{align}
For sufficiently small $K/N$, the right-hand side of inequality \eqref{geoseries} is less than 2, and thus $\sum_{m=0}^K{N\choose m}< 2{N\choose K}$ if $K\ll N$. We apply this inequality to equation \eqref{perimeter-eq}. Choose $n$ large enough so that $n\gamma_n<\frac{1}{2}\min\{|C_1|, |C_2|\}$, which is possible since $\gamma_n\rightarrow0$. Then for fixed $k\leq n\gamma_n$, we have that $|\partial N(C_1,k)| < 2\left[{|C_2|\choose k} + {|C_1|\choose k}\right]$ for large enough $n$. By another application of the inequality derived from \eqref{geoseries}, using the fact that $n\gamma_n = o(n)$, we therefore obtain
\begin{align*}
|N(C_1, \gamma_n)| = \sum_{k = 0}^{\lfloor n\gamma_n\rfloor} |\partial N(C_1, k)| & < \sum_{k = 0}^{\lfloor n\gamma_n\rfloor} 2\left\{{|C_2|\choose k} + {|C_1|\choose k}\right\}\\[1em]
& < 4\left\{{|C_2|\choose \lfloor n\gamma_n\rfloor} + {|C_1|\choose \lfloor n\gamma_n\rfloor}\right\}\leq 8{n\choose \lfloor n\gamma_n\rfloor}.
\end{align*}
As ${N\choose K}\leq\left(\frac{N\cdot e}{K}\right)^K$, we have
\[
|N(C_1, \gamma_n)| \leq \exp\left(\log(8) + n\gamma_n\left\{\log(e) + \log(1 / \gamma_n)\right\}\right) \leq \exp(3n\gamma_n\log(1/\gamma_n))
\]
for large enough $n$, since $1/\gamma_n\rightarrow\infty$.
\end{proof}




Here we give a short Corollary to Lemma \ref{fact:close-sets-complexity} which directly serves the proof of Lemma \ref{conc-cor}. Recall $\SmallNa$ from Definition \ref{N-ast} in Section \ref{thm2-lemmas}. 
\begin{corollary}\label{special_set_size}
Fix an integer $k > 1$. For large enough $n$,
\[
|\SmallNa|\leq 2\exp\left[3A\cdot n \cdot b_{n,k-1} \log\left(1/b_{n,k-1}\right)\right].
\]
\end{corollary}
\begin{proof} The corollary follows from a direct application of Lemma \ref{fact:close-sets-complexity} to $N(C_1, A\cdot b_{n,k-1})$ and $N(C_2, A\cdot b_{n,k-1})$.
\end{proof}




\begin{lemma}\label{pos-sum-fact}
Let $x_1, \ldots,x_k\in(0,1)$ be fixed and let $X_1,\ldots,X_k$ be arbitrary random variables. Define $W:=\sum_iX_i$ and $w:=\sum_ix_i$. Then for $t$ sufficiently small, $\mathbb{P}(|W_+^2 - w^2|>4k^2t) \leq \mathbb{P}(\max_i|X_i - x_i|>t)$.
\end{lemma}

\begin{proof}
Define $D_i := |X_i - x_i|$ and fix $t<\min_ix_i$. Then if $\max_iD_i\leq t$, all $X_i$'s will be positive, and thus $W_+ = W$ and $|W-w|\leq kt$, by the triangle inequality. Therefore $\max_i D_i\leq t$ implies that 
\begin{equation}\label{S-ineq}
|W_+^2 - w^2| = |(W-w)^2 + 2w(W-w)|\leq k^2t^2 + 2wkt\leq k^2t^2 + 2k^2t.
\end{equation}
Thus by the law of total probability, we have
\[
\mathbb{P}(|W_+^2 - w^2|>4k^2t) \leq \mathbb{P}(\{|W_+^2 - w^2|>4k^2t\}\cap\{\max_i D_i\leq t\}) + \mathbb{P}(\max_iD_i>t).
\]
Inequality \eqref{S-ineq} implies that for sufficiently small $t$, the first probability on the right-hand side above is equal to 0. The result follows.
\end{proof}

In what follows we state and prove Lemma \ref{lem:sup-score-conc-ineq0}, a concentration inequality for the modularity of a node set (see Definition \ref{set-modularity}) from a single-layer SBM with $n$ nodes and two communities. We first give a few short facts about the 2-community SBM. For for all results that follow, let $s, \rho$, and $v$ (see Definition \ref{rho}) correspond to the fixed set $B\subseteq[n]$ in each result (though sometimes we will make explicit the dependence on $B$). Define a matrix $V$ by $V(i,j) := P(i,j)(1 - P(i,j))$ for $i = 1, 2$, where $P$ is the probability matrix associated with the 2-block SBM.




\begin{lemma}\label{expv-YBDB}
Consider a single-layer SBM with $n>1$ nodes, two communities, and parameters $P$ and $\pi_1$. Fix a node set $B\subseteq[n]$ with $|B|\geq\alpha n$ for some $\alpha\in(0, 1)$. Then
\begin{enumerate}
\item $\left|\expv(Y(B)) - {|B|\choose 2}v^tPv\right| \leq 3|B|/2$
\item $\left|\expv\left(\sum_{u\in B}\widehat d(u)\right) - |B|nv^tP\pi\right| \leq |B|$
\item $\text{Var}\left(\sum_{u\in B}\widehat d(u)\right) \leq 9|B|n$
\end{enumerate}
\end{lemma}

\begin{proof} For part 1, note that by definition,
\(\expv\left(Y(B)\right) = \sum_{u,v\in B:u<v}\mathbb{P}\left((u,v)\in\widehat E\right) = \frac{1}{2}\sum_{u\ne v:u,v\in B}\mathbb{P}\left((u,v)\in\widehat E\right)\)
The right-hand sum can be expressed the sum of the entries of a $2\times 2$ symmetric block matrix with zeroes on the diagonal. In this matrix, the upper diagonal block is of size $|B\cap C_1|$ with off-diagonal entries equal to $P(1,1)$. The lower diagonal block is of size $|B\cap C_2|$ with off-diagonal entries equal to $P(2,2)$. The off-diagonal blocks have entries equal to $P(1,2)$. Therefore, summing over blocks and accounting for the zero diagonal, we have
\begin{align*}
\expv\left(Y(B)\right) & = \frac{1}{2}\Big(|B\cap C_1|^2P(1,1) + |B\cap C_1||B\cap C_2|P(1,2) + |B\cap C_2|^2 P(2,2)\Big)\\[1em]
&\;\;\; - \frac{1}{2}\Big(|B\cap C_1|P(1,1) +|B\cap C_2|P(2,2)\Big).
\end{align*}
By dividing and multiplying by $|B|^2$ and collapsing cross-products, we get
\begin{align*}
\expv\left(Y(B)\right) & = \frac{|B|^2}{2}\left(v^tPv - \frac{\rho P(1,1) + (1 - \rho) P(2,2)}{|B|}\right)\\[1em]
& = {|B|\choose 2}\left(1 + \frac{1}{|B|-1}\right)\left(v^tPv - \frac{\rho P(1,1) + (1 - \rho) P(2,2)}{|B|}\right)\\[1em]
& = {|B|\choose 2}\left(v^tPv - \frac{\rho P(1,1) + (1 - \rho) P(1,2)}{|B|} + \frac{v^tPv}{|B| - 1} - \frac{\rho P(1,1) + (1 - \rho) P(2,2)}{|B|(|B| - 1)}\right).
\end{align*}
Part 1 follows by carrying out the multiplication by ${|B|\choose 2}$ in the last expression.

For part 2, let $P(\cdot,i)$ denote the $i$-th column of $P$. Note that $\expv \left(\widehat{d}(u)\right) = n\pi^TP(\cdot, c_u) - P(c_u,c_u)$, with $c_u\in\{1,2\}$ denoting the community index of $u$. Therefore,
\begin{align*}
\expv\left(\sum_{u\in B}\widehat d(u)\right) & = \sum_{u\in B}\expv\left(\widehat d(u)\right) = \sum_{u\in B\cap C_1}\expv\left(\widehat d(u)\right) + \sum_{u\in B\cap C_2}\expv\left(\widehat d(u)\right)\\[1em]
& = |B\cap C_1|\left(n\pi^TP(\cdot, 1) - P(1,1)\right) + |B\cap C_2|\left(n\pi^TP(\cdot, 2) - P(2,2)\right)\\[1em]
& = |B|\left(n\rho\pi^TP(\cdot,1) + n(1-\rho)\pi^TP(\cdot, 2) - \rho P(1,1) - (1 - \rho) P(2,2) \right)\\[1em]
& = |B|nv^tP\pi - |B|\big(\rho P(1,1) + (1 - \rho) P(2,2)\big),
\end{align*}
which completes part 2.

Finally, for part 3, we have
\begin{equation}\label{sbm-fact-3}
\text{Var}\left(\sum_{u\in B}\widehat d(u)\right) = \text{Var}\left(2Y(B)\right) + \sum_{u,v:u\in B, v\in B^C}\text{Var}\left(\widehat X(u,v)\right).
\end{equation}
We address these two terms separately. For the first term, a calculation analogous to that from part 1 yields that $|\text{Var}\left(Y(B)\right) - {|B|\choose 2}v^tVv| \leq 3|B|/2$. Defining $\bar v:=(\rho(B^C),1 - \rho(B^C))^t$, it is easy to show that $\sum_{u,v:u\in B, v\in B^C}\text{Var}\left(\widehat X(u,v)\right) = |B||B^C|v^tV\bar v$, which is simply the sum of variances of all edge indicators for edges from $B$ to $B^C$. Applying these observations to equation \eqref{sbm-fact-3}, we have
\begin{align*}
\text{Var}\left(\sum_{u\in B}\widehat d(u)\right) & \leq 4{|B|\choose 2}v^tVv +  12|B|/2 + |B||B^C|v^TV\bar v\\[1em]
& \leq |B|\big(2(|B|-1)v^tVv + 6 + |B^C|v^tV\bar v\big)\leq9|B|n.
\end{align*}
\end{proof}




\begin{lemma}\label{E-conc}
Under a single-layer SBM with $n>1$ nodes, two communities, and parameters $P$ and $\pi_1$, define $\kappa:=\pi^TP\pi$. Then for large enough $n$, $\mathbb{P}\left(\big|2|\widehat E| - n^2\kappa\big| > t + 4n\right)\leq 2\exp\left\{-\frac{t^2}{n^2}\right\}$ for any $t>0$.
\end{lemma}

\begin{proof} Note that $|\widehat E| = Y([n])$. Thus part 1 of Lemma \ref{expv-YBDB} with $B = [n]$ yields $\left|\expv(|\widehat E|) - {n\choose 2}\kappa\right| \leq 3n/2$ for large enough $n$. As $n^2/2 = {n\choose 2} + n/2$, by the triangle inequality,
\(
\left|\expv(|\widehat E|) - \frac{n^2}{2}\kappa\right| \leq \left|\expv(|\widehat E|) - {n\choose 2}\kappa\right| + \frac{n}{2}\leq 2n
\)
Thus for any $t>0$, Hoeffding's inequality gives
\begin{align*}
\mathbb{P}\left(\left|2\widehat E - n^2\kappa\right|>t + 4n\right) & \leq \mathbb{P}\left(\left|2\widehat E - n^2\kappa\right|>t + 2\left|\expv(|\widehat E|) - \frac{n^2}{2}\kappa\right|\right)\\[1em]
& \leq \mathbb{P}\left(\left|2\widehat E - 2\expv(|\widehat E|)\right|>t\right)\leq 2\exp\left\{-2\frac{t^2}{4{n\choose 2}}\right\}\leq 2\exp\{-t^2/n^2\}.
\end{align*}
\end{proof}




\begin{lemma}
	\label{lem:sup-score-conc-ineq0}
Consider a single-layer 2-block SBM having $n>1$ nodes and parameters $P$ and $\pi$. Fix $\alpha \in(0,1)$ and $B\subseteq[n]$ such that $|B|\geq \alpha n$. Then for large enough $n$ we have
	\begin{equation}\label{eq:conc-sqrt}
		\mathbb{P}_n\left(\Big|\widehat{Q}(B) - \cQ(B)\Big| > \dfrac{t}{n^2} + \frac{8}{\kappa n}\right) \leq 4 \exp\left(-\dfrac{\kappa^2\alpha t^2}{16n^2}\right)
		\end{equation} 
for any $t>0$.

	\end{lemma}
\begin{proof}
With notation laid out in Section \ref{subsec:score}, define
\begin{equation}\label{proxy-score}
\widetilde Q(B) := n^{-1}{|B|\choose 2}^{-1/2}\left(Y(B) - \widetilde \mu(B)\right),
\end{equation}
where 
\begin{equation}\label{proxy-score-mu}
\widetilde \mu(B) := \frac{\underset{u,v \in B: u<v}\sum\widehat{d}(u)\widehat{d}(v)}{n^2 \kappa}.
\end{equation}
We will prove the inequality in three steps: \emph{Step 1}: bounding $\big|\widehat Q(B) - \widetilde{Q}(B)\big|$ in probability; \emph{Step 2}: deriving a concentration inequality for  $\widetilde{Q}(B)$; and \emph{Step 3}: showing that $\left|\expv\left(\widetilde{Q}(B)\right)-\cQ(B)\right|$ is eventually bounded by a constant.




\noindent \emph{Step 1.} As $\sum_{u,v\in B;u<v}\widehat d(u)\widehat d(v)\leq\sum_{u\in B}\widehat d(u)^2$, we have
\begin{align*}
\sum_{u,v \in B; u<v}\widehat{d}(u)\widehat{d}(v) & \leq \sqrt{\sum_{u,v \in B; u<v}\widehat{d}(u)\widehat{d}(v)}\sqrt{\sum_{u \in B}\widehat{d}(u)^2} \leq n{|B|\choose 2}^{1/2}|\widehat E|.
\end{align*}
Therefore,
\begin{align*}
\Big|{\widehat{Q}(B)} - {\widetilde{Q}(B)}\Big| & = n^{-1}{|B|\choose 2}^{-1/2}\Big|\dfrac{(2|\widehat E| - n^2\kappa) \sum_{u,v \in B; u<v}\widehat{d}(u)\widehat{d}(v)}{2|\widehat E|n^2\kappa }\Big| \leq \dfrac{\big|2|\widehat E| - n^2\kappa\big|}{2n^2 \kappa}.\addtocounter{equation}{1}\tag{\theequation}\label{step1-align1}
	\end{align*}
Combining the inequality in \eqref{step1-align1} with Lemma \ref{E-conc}, for any $t>0$,
\begin{align*}
\mathbb{P}\left(\Big|\widehat{Q}(B) - \widetilde{Q}(B) \Big| > \frac{t}{2n^2} + \frac{2}{\kappa n}\right) & \leq \mathbb{P}\left(\Big|2|\widehat E| - n^2\kappa\Big| > \kappa t + 4n\right)\nonumber \\[1em]
& \leq 2\exp\left(-\dfrac{\kappa^2 t^2}{n^2}\right). \addtocounter{equation}{1}\tag{\theequation}\label{step1-final}
\end{align*}




\noindent \emph{Step 2.} This step relies on McDiarmid's concentration inequality. Recall from Section \ref{sec:null} that $\widehat X(u,v)$ denotes the indicator of edge presence between nodes $u$ and $v$. Note that node pairs have a natural, unique ordering along the upper-diagonal of the adjacency matrix. Define $ord\{u,v\} = 2(u - 1) + (v - 1)$, for $\{u,v\}\in[n]^2$ with $u<v$ (e.g.\ $ord\{1,2\} = 1$, $ord\{1,3\} = 2$, etc.). For all $n>1$ and $i\leq n(n-1)/2$, define $\widehat Z(i) := \widehat X(u,v)$ such that $ord\{u,v\} = i$. If $ord\{u,v\} = i$, we call $\{u,v\}$ the ``$i$-th ordered node pair". Define the set
\(
\cI(B):=\{i:\text{the $i$-th ordered node pair has at least one node in $B$}\}
\)
and let $\widehat\cZ(B):=\{\widehat Z(i):i\in I(B)\}$. Note that the proxy score $\widetilde Q(B)$ is a function $f(z_1,z_2,\ldots)$ of the indicators $\widehat \cZ(B)$. 

Consider a \emph{fixed} indicator set $\cZ(B)$. For each $j\in I(B)$, define $\cZ^j(B) := \{Z^j(i):i\in I(B)\}$ with
\begin{equation}\label{ZjB}
Z^j(B) := \begin{cases}Z^j(i) = 1 - Z(i),&i = j\\Z^j(i) = Z(i), i\ne j\end{cases}
\end{equation}
To apply McDiarmid's inequality, we must bound $\Delta(j):=|f(\cZ(B)) - f(\cZ^{j}(B))|$ uniformly over $j\in \cI(B)$. Fix $j\in \cI(B)$ and let $\{u',v'\}$ be the $j$-th ordered edge. Without loss of generality, we assume $Z(j) = 1$. Since $f(\cZ(B)) = Q(B)$, $f(\cZ(B))$ has a representation in terms of $Y(B)$ and $\widetilde\mu(B)$. We let $Y^j(B)$ and $\widetilde\mu^j(B)$ correspond to $f(\cZ(B)^j)$. Notice that
\begin{equation}\label{step2-1}
n{|B|\choose 2}^{1/2}\Delta(j) = \left|Y(B) - Y^j(B) - \big(\widetilde \mu(B) - \widetilde \mu^j(B)\big)\right|.
\end{equation}
We bound the right hand side of equation \eqref{step2-1} in two cases: (i) $u',v'\in B$, and (ii) $u'\notin B, v'\in B$. In case (i), $Y(B) - Y^j(B) = 1$, and
\begin{align*}
\widetilde \mu(B) - \widetilde \mu^i(B) & = \dfrac{\displaystyle\sum_{u,v \in B; u\ne v}{d}(u){d}(v) - {d}^j(u){d}^j(v)}{n^2\kappa} = \dfrac{{d}(u'){d}(v')-{d}^j(u'){d}^j(v')}{n^2\kappa}\\[1em] & = \dfrac{{d}(u'){d}(v')-({d}(u') - 1)({d}(v') - 1)}{n^2\kappa} = \dfrac{{d}(u') + {d}(v')-1}{n^2\kappa},
\end{align*}
which is bounded in the interval $(0, 1)$ for large enough $n$. Thus in case (i), $\Delta(j) \leq 2{|B|\choose 2}^{-1/2}$ by the triangle inequality, for large enough $n$. In case (ii), $Y(B) - Y'(B) = 0$, and
\begin{align*}
\widetilde \mu(B) - \widetilde\mu^j(B) & = \dfrac{\displaystyle\sum_{u,v \in B; u\ne v}{d}(u){d}(v) - {d}^j(u){d}^j(v)}{n^2\kappa} = \dfrac{\displaystyle\sum_{u \in B; u\ne v'}{d}(u)\left({d}(v') - {d}^j(v')\right)}{n^2\kappa}\\[1em]
& = \dfrac{\displaystyle\sum_{u \in B; u\ne v'}{d}(u)}{n^2\kappa} \leq \dfrac{n|B|}{n^2\kappa}\leq \kappa^{-1}.
\end{align*}
Hence due to equation \eqref{step2-1}, we have for sufficiently large $n$ that
\begin{align*}
\Delta(j) & \leq n^{-1}{|B|\choose 2}^{-1/2}\cdot\max\{2, \kappa^{-1}\}\leq n^{-1}{|B|\choose 2}^{-1/2}\cdot2\cdot\kappa^{-1} \addtocounter{equation}{1}\tag{\theequation}\label{step2-align1}
\end{align*}
for all $j\in \cI(B)$, as $\kappa\leq1$. Since $|\cI(B)|= {|B|\choose 2} + |B||B^C|\leq n|B|$, McDiarmid's bounded-difference inequality implies that for sufficiently large $n$,
\begin{align*}
\mathbb{P}\left(\Big|\widetilde{Q}(B) - \expv\left(\widetilde{Q}(B)\right)\Big| > \frac{t}{n} \right) & = 2\exp\left(\frac{-t^2}{n|B|\Delta(j)}\right)\leq 2 \exp\left(-\kappa^2 \dfrac{n^2{|B|\choose 2} t^2}{4n^3|B|}\right) \\[1em]
& \leq 2 \exp\left(-\kappa^2 \dfrac{(|B| - 1) t^2}{8n}\right)\leq 2 \exp\left(-\kappa^2 \dfrac{\alpha t^2}{16}\right)
\end{align*}
for any $t>0$. Replacing $t$ by $t/n$ gives
\begin{equation}\label{step2-final}
\mathbb{P}\left(\Big|\widetilde{Q}(B) - \expv\left(\widetilde{Q}(B)\right)\Big| > \frac{t}{n^2} \right) \leq2\exp\left(-\kappa^2\dfrac{\alpha t^2}{16n^2}\right).
\end{equation}




\noindent \emph{Step 3.} Turning our attention to $\expv(\widetilde Q(B))$, recall that $n{|B|\choose 2}^{1/2}\widetilde Q(B) = Y(B) - \widetilde\mu(B)$ and that $\widetilde \mu(B) := \sum_{u,v\in B;u<v}\widehat d(u)\widehat d(v)/(n^2\kappa)$. As in previous lemmas, we will shorthand the quantities $s(B), \rho(B)$, and $v(B)$, by $s, \rho$, and $v$ (respectively). Note that
\begin{align*}
\expv\left(2\cdot\displaystyle\sum_{u,v\in B;u<v}\widehat d(u)\widehat d(v)\right) & = \expv\left\{\left(\sum_{u\in B}\widehat d(u)\right)^2 - \displaystyle\sum_{u\in B}\widehat d^2(u)\right\}\\[1em]
& = \text{Var}\left(\sum_{u\in B}\widehat d(u)\right) + \expv\left(\sum_{u\in B}\widehat d(u)\right)^2 - \displaystyle\sum_{u\in B}\expv\left(\widehat d^2(u)\right).\addtocounter{equation}{1}\tag{\theequation}\label{step3-1}
\end{align*}
Part 3 of Lemma \ref{expv-YBDB} gives $\text{Var}\left(\sum_{u\in B}\widehat d(u)\right)\leq 9sn^2$. Furthermore, for $u\in C_i$ we have
\begin{align*}
\expv\left(\widehat d^2(u)\right) & = \text{Var}\left(\widehat d(u)\right) + \expv\left(\widehat d(u)\right)^2\\[1em]
& = n\pi^TV(\cdot, i) - V(i, i) + n^2\left(\pi^TP(\cdot, i) - P(i,i)\right)^2,
\end{align*}
and therefore $\sum_{u\in B}\expv\left(\widehat d^2(u)\right) \leq 2sn^3$. Finally, Part 2 of Lemma \ref{expv-YBDB} gives\\
\noindent $\left|\expv\left(\sum_{u\in B}\widehat d(u)\right) - |B|nv^TP\pi\right| \leq |B|$. By expansion, this implies there exists a constant $a$ with $|a|<3$ such that for large enough $n$, $\expv\left(\sum_{u\in B}\widehat d(u)\right)^2 = s^2n^4(v^tP\pi)^2 + as^2n^3$. Therefore overall, line \eqref{step3-1} implies there exists a constant $b$ with $|b|<6$ such that for large enough $n$, $\expv\left(2\cdot\sum_{u,v\in B;u<v}\widehat d(u)\widehat d(v)\right) = s^2n^4(v^TP\pi)^2 + bsn^3$. Therefore, using the definition of $\widetilde \mu(B)$, 
\begin{align*}
\expv\left(\widetilde \mu(B)\right) & = s^2n^4\dfrac{(v^tP\pi)^2 + b(sn)^{-1}}{2n^2\kappa} 
 = {sn\choose 2}\left(1 + \frac{1}{sn - 1}\right)\left(\dfrac{(v^tP\pi)^2}{\kappa} + \frac{b}{\kappa sn}\right)\\[1em]
& = {sn\choose 2}\left(\dfrac{(v^tP\pi)^2}{\kappa} + \frac{b}{\kappa sn} + \frac{(v^tP\pi)^2}{\kappa(sn - 1)} + \frac{b}{\kappa sn(sn - 1)}\right)\\[1em]
& = {sn\choose 2}\left(\dfrac{(v^tP\pi)^2}{\kappa} + \frac{1}{\kappa sn}\left(b + \frac{sn(v^tP\pi)^2 + b}{sn - 1}\right)\right) \\[1em]
& = {sn\choose 2}\left(\dfrac{(v^tP\pi)^2}{\kappa} + \frac{c_1}{\kappa sn}\right)\addtocounter{equation}{1}\tag{\theequation}\label{eq:last1-2}
\end{align*}
for a constant $c_1$ with $|c_1|<8$, for large enough $n$. Now, part 1 of Lemma \ref{expv-YBDB} gives that $\left|\expv(Y(B)) - {|B|\choose 2}v^tPv\right| \leq 3|B|/2$ for large enough $n$. Thus there exists a constant $c_2$ with $|c_2|<3$ such that for large enough $n$, $\expv(Y(B)) = {|B|\choose 2}\left(v^tPv + \frac{c_2}{sn}\right)$. Thus
\begin{align*}
n\expv\left(\widetilde Q(B)\right) & = {|B|\choose 2}^{-1/2}(\expv (Y(B)) - \expv (\tilde\mu(B))) = \frac{sn}{\sqrt{2}}\left(v^tPv - \frac{(v^tP\pi)^2}{\kappa} + \frac{1}{sn}\left(c_1/\kappa + c_2\right)\right) \\[1em]
& *\left(\sqrt{1 - \frac{1}{sn}}\right) = \frac{sn}{\sqrt{2}}\left(v^tPv - \frac{(v^tP\pi)^2}{\kappa}\right) + \frac{c_1/\kappa + c_2}{\sqrt{2}}\left(\sqrt{1 - \frac{1}{sn}}\right).
\end{align*}
Thus there exists a constant $c$ with $|c|\leq |c_1|/\kappa + |c_2| < 8/\kappa + 3$ such that for large enough $n$, $\expv(\widetilde Q(B)) = \cQ(B) + c/n$. This completes Step 3.\\




\noindent {\bf Completion of the proof}: We now recall the results of the three steps:
\begin{enumerate}[(i)]
\item For large enough $n$, we have $\mathbb{P}\left(\Big|\widehat{Q}(B) - \widetilde{Q}(B) \Big| > \frac{t}{2n^2} + \frac{2}{\kappa n}\right)\leq 2\exp\left(-\dfrac{\kappa^2 t^2}{n^2}\right)$
\item $\mathbb{P}\left(\Big|\widetilde{Q}(B) - \expv\left(\widetilde{Q}(B)\right)\Big| > \frac{t}{n^2} \right) \leq2\exp\left(-\kappa^2\dfrac{\alpha t^2}{16n^2}\right)$
\item There exists $c$ with $|c|<8/\kappa + 3$ such that for large enough $n$, $\expv\left(\widehat Q(B)\right) = q(B) + c/n$ 
\end{enumerate}
Noting that $\alpha/16 < 1$, we apply a union bound to the results of steps (i) and (ii):
\begin{equation}
\mathbb{P}\left(\Big|\widehat{Q}(B) - \expv\left(\widetilde{Q}(B)\right) \Big| > \frac{t}{n^2} + \frac{2}{\kappa n}\right)\leq 4\exp\left(-\dfrac{\kappa^2\alpha t^2}{16n^2}\right)
\end{equation}
Applying the inequality $|x-a|\geq |x|-|a|$ with (iii) and some algebra gives the result.\end{proof}

\section{Competing Methods}
In Sections \ref{sec:application} and \ref{sec:simulations}, we compare and contrast the performance of Multilayer Extraction with the following methods:

\begin{enumerate}
	\item[]{\it Spectral clustering \citep{newman2006finding}}: an iterative algorithm based on the spectral properties of the modularity matrix of an observed network. In the first step, the modularity matrix of the observed network is calculated and its leading eigenvector is identified. The graph is divided into two disjoint communities so that each vertex is assigned according to its sign in the leading eigenvector. Next, the modularity matrix is calculated for both of the subgraphs corresponding to the previous division. If the modularity of the partition increases, these communities are once again divided into two disjoint communities, and the procedure is repeated in this fashion until the modularity no longer increases. For the desired igraph object \texttt{graph}, the call for this in \texttt{R} was: 
	{\small 
	\begin{verbatim} cluster_leading_eigen(graph, steps = -1, weights = NULL, start = NULL, 
		options = arpack_defaults, callback = NULL, extra = NULL, env = parent.frame()) \end{verbatim}
		}
	\item[]{\it Label Propagation \citep{raghavan2007near}}: an iterative algorithm based on propagation through the network. At the first step, all vertices are randomly assigned a community label. Sequentially, the algorithm chooses a single vertex and updates the labels of its neighborhood to be the majority label of the neighborhood. The algorithm continues updating labels in this way until no updates are available. For the desired igraph object \texttt{graph}, the call for this in \texttt{R} was: 
	{\small 
	\begin{verbatim} cluster_label_prop(graph, weights = E(graph)$weight, initial = NULL, 
		fixed = NULL) \end{verbatim}
		}
	\item[]{\it Fast and greedy \citep{clauset2004finding}}: an iterative and greedy algorithm that seeks a partition of vertices with maximum modularity. The algorithm is an agglomerative approach that is a modification of the Kernighan-Lin algorithm commonly used in the identification of community structure in network. For the desired igraph object \texttt{graph}, the call for this in \texttt{R} was: 
	{\small
	\begin{verbatim} cluster_fast_greedy(graph, merges = TRUE, weights = E(graph)$weight, 
		modularity = TRUE, membership = TRUE, weights = NULL)
  \end{verbatim}
  }
	\item[]{\it Walktrap \citep{pons2005computing}}: an agglomerative algorithm that seeks a partition of vertices that minimizes the total length of a random walk within each community. At the first stage, each vertex of the network is placed in its own community. At each subsequent stage, the two closest communities (according to walk distance) are merged. This process is continued until all vertices have been merged into one large community, and a community dendrogram is formed. The partition with the smallest random walk distance is chosen as the final partition.   
	{\small \begin{verbatim} cluster_walktrap(graph, weights = E(graph)$weight, steps = 4,
  merges = TRUE, modularity = TRUE, membership = TRUE)
  	\end{verbatim}
	}
	\item[]{\it GenLouvain \citep{jutla2011generalized}}: a multilayer generalization of the iterative GenLouvain algorithm. This algorithm seeks a partition of the vertices and layers that maximizes the multilayer modularity of the network, as described in \citep{mucha2010community}. We use the MATLAB implementation from \citep{jutla2011generalized} of GenLouvain with resolution parameter set to 1, and argument \texttt{randmove = moverandw}. 
	\item[]{\it Infomap \citep{de2014identifying}}: a multilayer generalization of the Infomap algorithm from \citet{rosvall2008maps}. This algorithm seeks to identify a partition of the vertices and layers that minimize the generalized map equation, which measures the description length of a random walk on the partition. We use the C++ multiplex implementation of Infomap provided at \url{http://www.mapequation.org/code.html}. In implementation, we set the arguments of the function to permit overlapping communities, and set the algorithm to ignore self-loops. 
\end{enumerate}

For the first four methods, we use the default settings from the {\it igraph} package version 0.7.1 set in \texttt{R}.  
\section{Extraction Simulations}\label{sec:appendix-sims}
\subsection{Simulation}
We now investigate several intrinsic properties of Multilayer Extraction by applying the method to multilayer networks with several types of community structure, including I) disjoint, II) overlapping, III) persistent, IV) non-persistent, and V) hierarchical structure. Figure \ref{fig:ext.sims} illustrates six multilayer networks that we analyze for this purpose. Each simulated network contains 1000 nodes and 90 layers. Embedded communities have inner connection probability 0.15; whereas, the remaining vertices independently connected to all other vertices with probability 0.05.

\begin{figure}[ht]
\centering
\includegraphics[width = 0.75\textwidth, trim = 1.5cm 8cm 1.5cm 1.5cm, clip = TRUE]{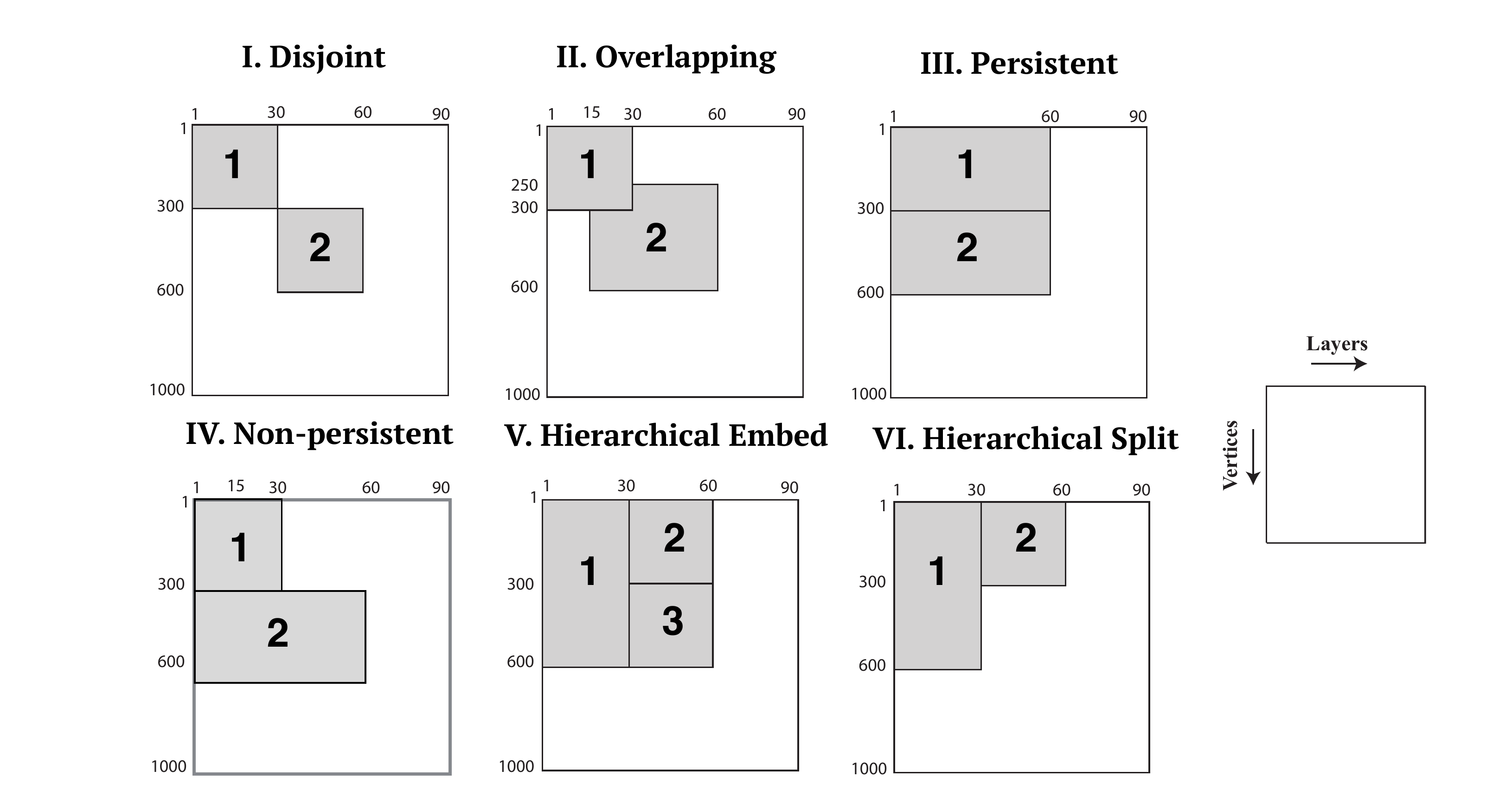}
\caption[Extraction Simulation Test Bed]{\small Simulation test bed for extraction procedures. Each graphic displays a multilayer network on 1000 nodes and 90 layers. In each plot, shaded rectangles are placed over the nodes (rows) and layers (columns) that are included in a multilayer community. Communities are labeled by number. Vertices within the same community are randomly connected with probability 0.15 while all other vertices have connection probability 0.05 to vertices in their respective layer.}
\label{fig:ext.sims}
\end{figure}

\subsection{Results}
In the disjoint, overlapping, persistent, and non-persistent networks (I, II, III, and IV, respectively), Multilayer Extraction identifies communities that perfectly match the true embedded communities. On the other hand, in the hierarchical community setting, Multilayer Extraction is unable to identify the full set of communities. In example V, Multilayer Extraction does not identify community 1, and in example VI Extraction identifies a community with vertices 1 - 300 across layers 1 - 60, which combines community 1 and community 2. 

Together, these results suggest two properties of the Multilayer Extraction procedure. First, the method can efficiently identify disjoint and overlapping community structure in multilayer networks with heterogeneous community structure. Second, Multilayer Extraction tends to disregard communities with a large number of vertices (e.g. communities that include over half of the vertices in a network). The inverse relationship between the score and the number of vertices in a community may provide some justification as to why this is the case. In networks with large communities, one can in principle modify the score by introducing a reward for large collections. We plan to pursue this further in future research.

\singlespacing
\bibliographystyle{Chicago}

\end{document}